\documentclass[11pt,a4paper,centertags,reqno,twoside]{amsart}
\usepackage{amsmath,amssymb,latexsym, graphicx}
\usepackage{mathptmx}
\usepackage{color}
\usepackage{amsfonts}
\usepackage{amsthm}
\usepackage[mathcal]{euscript}
\usepackage{bm}
\usepackage[english]{babel}
\usepackage[bookmarks]{hyperref}
\numberwithin{equation}{section}

\renewcommand{\epsilon}{\varepsilon}

\renewcommand{\epsilon}{\varepsilon}

\renewcommand{\d}{\mathrm{d}}
\renewcommand{\hat}{\widehat }

\newcommand{\be}{\begin{equation}}
\newcommand{\ee}{\end{equation}}

\newcommand{\C}{\mathbb{C}}

\newcommand{\N}{\mathbb{N}}

\newcommand{\R}{\mathbb{R}}

\newcommand{\T}{\mathbb{T}}

\renewcommand{\H}{\mathbb{H}}
\newcommand{\V}{\mathbb{V}}

\newcommand{\Z}{\mathbb{Z}}
\newcommand{\cA}{{\mathcal A}}

\newcommand{\cC}{{\mathcal C}}

\newcommand{\cE}{{\mathcal E}}
\newcommand{\cF}{{\mathcal F}}
\newcommand{\cG}{{\mathcal G}}

\newcommand{\cK}{{\mathcal K}}

\newcommand{\cM}{{\mathcal M}}

\newcommand{\cS}{{\mathcal S}}

\renewcommand{\det}{\mathop{\mathrm{det}}}

{\bf}{\it}
\newtheorem{theorem}{Theorem}[section]
\newtheorem{lemma}[theorem]{Lemma}
\newtheorem{corollary}[theorem]{Corollary}

\newtheorem{proposition}[theorem]{Proposition}

\theoremstyle{remark}
\newtheorem{remark}[theorem]{Remark}

\date{\today}

\begin{document}

\large
\title[Two-boson lattice Hamiltonian]{A two-boson lattice Hamiltonian
with interactions up to next-neighboring sites}


\author{Saidakhmat N.~Lakaev, Alexander K. Motovilov,  Mukhayyo O.~Akhmadova}

\address[Saidakhmat Lakaev]{Samarkand State University, 140104, Samarkand, Uzbekistan}
\email{slakaev@mail.ru}

\address[Alexander K. Motovilov]{Bogoliubov Laboratory of
Theoretical Physics, JINR, Joliot-Cu\-rie 6, 141980 Dubna, Russia,
and Dubna State University, Universitetskaya 19, 141980 Dubna,
Russia} \email{motovilv@theor.jinr.ru}

\address[Mukhayyo Akhmadova]{Samarkand State University, 140104, Samarkand, Uzbekistan}
\email{mukhayyo@mail.ru}

\begin{abstract}
A system of two  identical spinless bosons on the two-dimensional
lattice is considered under the assumption that on-site and first
and second nearest-neighboring site interactions between the bosons
are only nontrivial  and that these interactions are of magnitudes
$\gamma$, $\lambda$, and $\mu$, respectively. A partition of the
$(\gamma,\lambda,\mu)$-space into connected components is
established such that, in each connected component, the two-boson
Schr\"odinger operator corresponding to the zero quasi-momentum of
the center of mass has a definite (fixed) number of eigenvalues,
which are situated below the bottom of the essential (continuous)
spectrum and above its top. Moreover,  for each connected component,
a sharp lower bound is established on the number of isolated
eigenvalues for the two-boson Schr\"odinger operator corresponding
to any admissible nonzero value of the center-of-mass quasimomentum.

\end{abstract}


\maketitle

\section{Introduction}

Lattice models of physical systems are widely used in mathematical
physics. One of such models is the lattice $N$-body Hamiltonian
\cite{Mattis:1986} that may be interpreted as the simplest version
of the corresponding Bose-Hubbard model involving just $N$ particles
of a certain type. The lattice $N$-particle Hamiltonians are of
great interest even in the case of $N\leq 3$ and the corresponding
lattice $N$-particle problems have been intensively studied over the
past several decades \cite{ALzM:2004,ALMM:2006,
ALKh:2012,BdsPL:2017,FICarroll:2002, HMumK:2020, KhLA:2021,
Lakaev:1993, LAbdukhakimov:2020, LBA:2022, LBozorov:2009,
LO'zdemir:2016,LKhL:2012,LMAbduhakimov:2023}. One more reason to
study the lattice $N$-particle Hamiltonians is that, in fact, they
represent a natural discrete approximation of the respective
original continuous $N$-body Schr\"odinger operators
\cite{Faddeev:1986}; setting down the $N$-body problem on a lattice
allows one to study it in the context of the theory of bounded
operators.

The lattice $N$-body Schr\"odinger operators also represent the
simplest model for systems formed by $N$ particles traveling through
periodic structures, e.g. for ultracold atoms injected into optical
crystals \cite{Bloch:2005, Winkler:2006}. The study of ultracold
few-atom systems in optical lattices have became popular in the last
years due to availability of controllable parameters such as
temperature, particle masses, interaction potentials etc. (see e.g.,
\cite{Bloch:2005,JBC:1998, JZ:2005, Lewenstein:2012} and references
therein).


It is well known that the celebrated Efimov effect
\cite{Efimov:1970} was originally attributed to a system of three
particles on the three-dimensional continuous space $\mathbb{R}^3$.
A rigorous mathematical proof of the Efimov effect has been given in
\cite{Ovchinnikov:1979,Sobolev:1993,Tamura:1991,Yafaev:1974}. Later
on it has been shown that the Efimov effect takes place also in
three-particle systems on lattices
\cite{ALzM:2004,ALKh:2012,DzMSh:2011,Lakaev:1993}. Thus, the lattice
three-body problems represent one more active field for Efimov
physics \cite{EfiPhys}.

Lattice Hamiltonians are of an interest in fusion physics, too. For
example, in \cite{Motovilov:2001}, a one-dimensional lattice based
Hamiltonian has been successfully employed to show that an
arrangement of molecules of a certain type in a lattice structure
may essentially enhance their nuclear fusion probability.

It is well known that in the continuous case, the center-of-mass
motion of a system of $N$ particles can be separated, but this is
not the case for the lattice $N$-particle problems. However, the
discrete translation invariance of a lattice Hamiltonian gives an
opportunity to employ the Floquet-Bloch decomposition (see, e.g.,
\cite[Sec. 4]{ALMM:2006}). In particular, for the $N$-particle
lattice Hamiltonian $\mathrm{H}$ acting on a functional Hilbert
space, defined on $\T^{Nd}$,  one can use the following  von Neumann
direct integral decomposition
\begin{equation}\label{entries}
\mathrm{H}\simeq\int\limits_{K\in \T^d} ^\oplus  H(K)\,d K,
\end{equation}
where $\T^d$ stands for the $d$-dimensional torus, $K$ for
quasimomentum of the center of mass, and $H(K)$  is the so called
fiber Hamiltonian. For each $K\in\T^d$, the latter is an operator in
the Hilbert space on $\T^{(N-1)d}$. The decomposition
\eqref{entries} allows one to reduce the problem of the study of the
total Hamiltonian $\mathrm{H}$ to studying the fiber operators
$H(K)$. We notice that the dependence of $H$ on the quasimomentum
$K\in\T^d$ is quite nontrivial (see e.g., \cite{ ALzM:2004,
ALMM:2006, RSimonIII:1982}).

In this work, we restrict ourselves to the case were $N=d=2$ and the
particles are identical spinless bosons. This implies that, for any
$K\in\T^2$, the fiber Hamiltonian $H(K)$ should be an operator in
the Hilbert space $L^{2,\mathrm{e}}(\T^2)$ whose elements are
measurable square integrable functions $f:\T^2\to\C$ which, in
addition, are even, that is $f(-p)=f(p)$, $p\in\T^2$.   More
precisely, we
 study the spectral properties of the fiber
Hamiltonians $H(K),\,K\in\T^2$, that act on $L^{2,\mathrm{e}}(\T^2)$
as
\begin{equation}\label{fiber-intro}
H_{\gamma\lambda\mu}(K):=H_0(K) + V_{\gamma\lambda\mu},
\end{equation}
where $H_0(K)$ is the fiber kinetic-energy operator,
$$
\bigl(H_0(K)f\bigr)(p)=\cE_K(p)f(p),
$$
with
\begin{equation}\label{def:dispersion}
\cE_K(p):= 2 \sum_{i=1}^2\Big(1-\cos\tfrac{K_i}2\,\cos p_i\Big)
\end{equation}
and $V_{\gamma\lambda\mu}$ is the combined interaction potential
that does not depend on $K$. The parameters  $\gamma$, $\lambda$ and
$\mu$, called coupling constants, are nothing but factors
determining interactions between the particles which are located on
same site, on nearest neighboring sites and on next-to-nearest
neighboring sites of the lattice, respectively.

As a next step, we realize that, in the case of the potential
$V_{\gamma\lambda\mu}$ employed, the bosonic space
$L^{2,\mathrm{e}}(\T^2)$ itself admits an orthogonal decomposition
into three subspaces:
$$
L^{2,{\rm e}}(\T^2)=L^{2,\mathrm{ees}}(\T^2) \oplus
L^{2,\mathrm{oos}}(\T^2) \oplus L^\mathrm{2,ea}(\T^2),
$$
where
\begin{align}\label{def:ees_space}
L^{2,\mathrm{ees}}&(\T^2)\\
&\nonumber
=\{f\in L^{2,\mathrm{e}}(\T^2): f(p_1,p_2)=f(p_2,p_1)=f(-p_1,p_2), \,\, \forall p_1,p_2\in\T\}\\
L^{2,\mathrm{oos}}&(\T^2)\label{def:oos_space} \\
\nonumber
&=\{f\in L^{2,\mathrm{e}}(\T^2):f(p_1,p_2)=f(p_2,p_1)=-f(-p_1,p_2), \,\, \forall p_1,p_2\in\T \}\\
\label{def:ea_space} L^{2,\mathrm{ea}}&(\T^2)\\
\nonumber &=\{{f}\in L^\mathrm{2,e}(\T^2): {f}(p_1,p_2)=
-{f}(p_2,p_1),\,\, \forall p_1,p_2\in\T\}.
\end{align}
The important thing is that the orthogonal subspaces
$L^{2,\mathrm{ees}}(\T^2)$,\,\,$L^{2,\mathrm{oos}}(\T^2)$, \, and \,
$L^{\rm 2,ea}(\T^2)$ are reducing for ${H}_{\gamma\lambda\mu}(0)$.
This allows one to study the spectrum of the Hamiltonian
${H}_{\gamma\lambda\mu}(0)$ separately for its parts in
$L^{2,\mathrm{ees}}(\T^2)$,\,\,$L^{2,\mathrm{oos}}(\T^2)$, \, and \,
$L^{\rm 2,ea}(\T^2)$.

In the model \eqref{fiber-intro}, we determine both the exact number
and location of the eigenvalues for the operator
$H_{\gamma\lambda\mu}(0)$, depending on the interaction parameter
$\gamma$, $\lambda$ and $\mu$. We also reveal the mechanisms of
emergence and absorption of these eigenvalues at the thresholds of
the essential (continuous) spectrum of $H_{\gamma\lambda\mu}(0)$, as
the parameters $\gamma$, $\lambda$ and $\mu$ vary. Furthermore, for
any value of the quasimomentum $K\in\T$ we establish sharp lower
bounds for the numbers of isolated eigenvalues of $H_{\gamma \lambda
\mu}(K)$ (counting their multiplicities) that lie below the
essential spectrum and above it.

The paper is organized as follows. In Sec. \ref{sec:hamiltonian} we
introduce the two-boson Hamiltonian in the position and
(quasi)momentum representations. Sec. \ref{sec:main_results}
contains statements of our main results. Some auxiliary technical
lemmas are presented in Sec. \ref{sec:auxiliary}. Proofs of the main
results are given in Sec. \ref{sec:proofs}.

\section{Hamiltonian of a two-boson system on the lattice $\Z^2$
and its basic properties} \label{sec:hamiltonian}

\subsection{The position-space representation}

Let $\Z^2:=\Z\times\Z$ be the two-dimensional integer lattice and
let $\ell^{2,s}(\Z^2\times\Z^2)$ be the Hilbert space of
square-summable symmetric functions on $\Z^2\times\Z^2.$ We consider
a system of two  identical spinless bosons on the two-dimensional
lattice under the assumption that on-site and first and second
nearest-neighboring site interactions between the particles are only
nontrivial  and that these interactions are of magnitudes $\gamma$,
$\lambda$, and $\mu$, respectively. The corresponding Hamiltionian
in the position-space represenatation is introduced on
$\ell^{2,s}(\Z^2\times\Z^2)$ as
\begin{equation}\label{two_total}
\hat \H_{\gamma\lambda\mu}=\hat \H_0 + \hat
\V_{\gamma\lambda\mu},\qquad \gamma,\lambda,\mu\in\R.
\end{equation}
Here, $\hat \H_0$ stands for the Hamiltonian of the pair of
identical non-interacting bosons:
\begin{equation*}
(\hat \H_0 \hat f)(x_1,x_2)= \sum_{y\in\Z^2} \hat \epsilon(x_1-y)
\hat f(y,x_2) + \sum\limits_{y\in\Z^2} \hat \epsilon(x_2-y) \hat
f(x_1,y),\quad x_1,x_2\in\Z^2,
\end{equation*}
where
$$
\hat \epsilon(s) =
\begin{cases}
2 & \text{if $|s|=0,$}\\
-\frac{1}2 & \text{if $|s|=1,$}\\
0 & \text{if $|s|>1,$}
\end{cases}
$$
and $|s|=|s_1|+|s_2|$ for $s=(s_1,s_2)\in \Z^2.$ The interaction
$\hat \V_{\gamma\lambda\mu}$ is the operator of multiplication by a
function $\hat{v}_{\gamma\lambda\mu}$,
\begin{equation*}\label{interaction}
(\hat \V_{\gamma\lambda\mu} \hat f) (x,y) = \hat
v_{\gamma\lambda\mu}(x-y) \hat f(x,y),
\end{equation*}
where
$$
\hat{v}_{\gamma\lambda\mu}(x) =
\begin{cases}
\gamma & |x| = 0,\\
\frac{\lambda}{2} & |x| = 1,\\
\frac{\mu}{2} & |x| = 2,\\
0 & |x|>2.
\end{cases}
$$

\subsection{The quasimomentum-space representation.}
By $\T^2$ we denote the two-dimensional torus,  $\T^2=(\R /2\pi
\Z)^2\equiv [-\pi,\pi)^2$.  The torus $\T^2$ coincides with the
Pontryagin dual group of $\Z^2$, equipped with the Haar measure $\d
p$. Let $L^{2,s}(\T^2\times\T^2)$ be the Hilbert space of
square-integrable symmetric functions on $\T^2\times\T^2.$ We use
the notation $\cF$ for the standard Fourier transform from
$\ell^2(\Z^2)$ to $L^2(\T^2)$,
$$
(\cF \hat f)(p)=\frac{1}{2\pi} \sum_{x\in\Z^2} \hat f(x) e^{ip\cdot
x},
$$
where $p\cdot x: = p_1x_1+p_2x_2$ for $p=(p_1,p_2)\in\T^2$ and
$x=(x_1,x_2)\in\Z^2.$


The quasimomentum-space version of the Hamiltonian \eqref{two_total}
reads as
$$
\H_{\gamma\lambda\mu}:=(\cF\otimes\cF) \hat
\H_{\gamma\lambda\mu}(\cF\otimes\cF)^*:=\H_0 +
\V_{\gamma\lambda\mu}.
$$
Here the free Hamiltonian  $ \H_0=(\cF \otimes \cF) \hat \H_0
(\cF\otimes \cF)^*$  is the multiplication operator:
$$
(\H_0 f)(p,q) = [\epsilon(p) + \epsilon(q)]f(p,q),
$$
where
$$
\epsilon(p) := \sum\limits_{i=1}^2 \big(1-\cos p_i),\quad
p=(p_1,p_2)\in \T^2,
$$
is the \emph{dispersion relation} of a single boson.
The interaction  $\V_{\gamma\lambda\mu}=(\cF \otimes
\cF)\hat\V_{\gamma\lambda\mu} (\cF\otimes\cF)^*$ is the (partial)
integral operator
$$
(\V_{\gamma\lambda\mu} f)(p,q) = \frac{1}{(2\pi)^2}\int_{\T^2}
v_{\gamma\lambda\mu}(p-u) f(u,p+q-u)\d u,
$$
where
$$
v_{\gamma\lambda\mu}(p)= \gamma+\lambda\sum_{i=1}^2\cos p_i + \mu
\sum_{i=1}^2\cos 2p_i+2\mu\cos p_1\cos p_2,\quad p=(p_1,p_2)\in
\T^2.
$$

\subsection{The Floquet-Bloch decomposition of $\H_{\gamma\lambda\mu}$
and fiber Schr\"odinger operators}\label{subsec:von_neuman}

Since $\hat H_{\gamma\lambda\mu}$ commutes with the representation
of the discrete group $\Z^2$ by shift operators on the lattice, we
can decompose the space $L^{2,s}(\T^2\times\T^2)$ and
$\H_{\gamma\lambda\mu}$ into the von Neumann direct integral as
\begin{equation}\label{hilbertfiber}
L^{2,s}(\T^2\times \T^2)\simeq \int\limits_{K\in \T^2} \oplus
L^{2,e}(\T^2)\,\d K
\end{equation}
and
\begin{equation}\label{fiber}
\H_{\gamma\lambda\mu} \simeq \int\limits_{K\in \T^2} \oplus
H_{\gamma\lambda\mu}(K)\,\d K,
\end{equation}
where $L^{2,e}(\T^2)$ is the Hilbert space of square-integrable even
functions on $\T^2$ (see, e.g., \cite{ALMM:2006}). The fiber
operator $H_{\gamma\lambda\mu}(K),$ $K\in\T^2,$ is a self-adjoint
operator on $L^{2,e}(\T^2)$ of the form
\begin{equation*}
H_{\gamma\lambda\mu}(K) := H_0(K) + V_{\gamma\lambda\mu},
\end{equation*}
where the unperturbed operator $H_0(K)$ is the operator of
multiplication by the function
\begin{equation}\label{multiplication}
\cE_K(p):= 2 \sum_{i=1}^2\Big(1-\cos\tfrac{K_i}2\,\cos p_i\Big).
\end{equation}
and the perturbation  $V_{\gamma\lambda\mu}$ reads as
\begin{align} \label{moment_poten}
V_{\gamma\lambda\mu}f(p) =
&\frac{\gamma}{4\pi^2}\int\limits_{\T^2}f(q)\d q +
\frac{\lambda}{4\pi^2}
\sum\limits_{i=1}^2 \cos p_i \int\limits_{\T^2} \cos q_i f(q) \d q   \\
\nonumber &\,+\frac{\mu}{4\pi^2} \sum\limits_{i=1}^2\cos 2p_i
\int\limits_{\T^2}\cos 2q_if(q)\d q \\
\nonumber &+\frac{\mu}{2\pi^2}\cos p_1 \cos p_2
\int\limits_{\T^2}\cos q_1 \cos q_2 f(q)\d q  \\
&+ \frac{\mu}{2\pi^2}\sin p_1 \sin p_2 \int\limits_{\T^2}\sin q_1
\sin q_2 f(q)\d q. \nonumber
\end{align}
The parameter $K\in\T^2$ is usually called the \emph{two-particle
quasimomentum} and the entry $H_{\gamma\lambda\mu}(K)$ is called the
\emph{fiber Schr\"odinger operator} associated with the two-particle
Hamiltonian $\hat \H_{\gamma\lambda\mu}$.

\subsection{Essential spectrum of the fiber Schr\"odinger operators}
\label{subsec:ess_spec}

Depending on $\gamma,\lambda, \mu \in \R$ the rank of
$V_{\gamma\lambda\mu}$  varies but never exceeds seven. Hence, by
the classical Weyl theorem, for any $K\in\T^2$, the essential
spectrum $\sigma_{\mathrm{ess}}(H_{\gamma\lambda\mu}(K))$  of the
fiber Schr\"odinger operator $H_{\gamma\lambda\mu}(K)$ coincides
with the spectrum of $H_0(K):$
\begin{equation}\label{eq:ess_spec}
\sigma_{\mathrm{ess}}(H_{\gamma\lambda\mu}(K)) = \sigma(H_0(K)) =
[\cE_{\min}(K),\cE_{\max}(K)],
\end{equation}
where
\begin{align}
\label{EKmin} \cE_{\min}(K):= & \min_{p\in  \T ^2}\,\cE_K(p) =
2\sum\limits_{i=1}^2\Big(1-\cos \tfrac{K_{i}}2\Big)\geq \cE_{\min}(0)=0,\\
\label{EKmax} \cE_{\max}(K):= & \max_{p\in  \T ^2}\,\cE_K(p)=
2\sum\limits_{i=1}^2\Big(1+\cos \tfrac{K_{i}}2\Big)\leq
\cE_{\max}(0)=8.
\end{align}

\subsection{Relationship between the discrete spectra of
$H_{\gamma\lambda\mu}(K)$ for $K\neq 0$ and
$H_{\gamma\lambda\mu}(0)$}\label{sec:Rel0K}

For every $n\ge1$ we define
\begin{equation}\label{enK}
e_n(K;\gamma,\lambda,\mu):= \sup\limits_{\phi_1,\ldots,\phi_{n-1}\in
L^{2,\rm e}(\T^2)}\,\,\inf\limits_{{\psi}
\in[\phi_1,\ldots,\phi_{n-1}]^\perp,\,\|{\psi}\|=1}
({H}_{\gamma\lambda\mu}(K){\psi},{\psi})
\end{equation}
and
\begin{equation}\label{EnK}
E_n(K; \gamma,\lambda,\mu):=
\inf\limits_{\phi_1,\ldots,\phi_{n-1}\in L^{2,\rm
e}(\T^2)}\,\,\sup\limits_{{\psi}
\in[\phi_1,\ldots,\phi_{n-1}]^\perp,\,\|{\psi}\|=1}
({H}_{\gamma\lambda\mu}(K){\psi},{\psi}).
\end{equation}
By the minimax principle, $e_n(K;\gamma,\lambda,\mu)\le
\cE_{\min}(K)$ and $E_n(K;\gamma,\lambda,\mu)\ge \cE_{\max}(K).$
From the representation \eqref{moment_poten} it follows that the
rank of $V_{\gamma\lambda\mu}$ does not exceed seven. Then, by
choosing suitable elements $\phi_1$, $\phi_2$, $\phi_3$, $\phi_4$,
$\phi_5$,$\phi_6$  and $\phi_7$ from the range of
$V_{\gamma\lambda\mu}$,  one concludes that
$e_n(K;\gamma,\lambda,\mu) = \cE_{\min}(K)$ and
$E_n(K;\gamma,\lambda,\mu) = \cE_{\max}(K)$ for all $n\ge8.$ If the
numbers $e_n(K;\gamma,\lambda,\mu)<\cE_{\min}(K)$ and/or
$E_n(K;\gamma, \lambda,\mu)>\cE_{\max}(K)$ exist, they represent the
discrete eigenvalues of $H_{\gamma\lambda\mu}(K)$, $K\in\T^2$.

\begin{lemma}\label{lem:monoton_xos_qiymat}
Let $n\ge1$ and $i\in\{1,2\}.$ For each fixed $K_j\in\T,$
$j\in\{1,2\}\setminus\{i\},$  the map
$$
K_i\in\T \mapsto \cE_{\min}((K_1,K_2)) -
e_n((K_1,K_2);\gamma,\lambda,\mu)
$$
is non-increasing in $(-\pi,0]$ and non-decreasing in $[0,\pi]$.
Similarly, for every fixed $K_j\in\T,$ $j\in\{1,2\}\setminus\{i\},$
the map
$$
K_i\in\T \mapsto E_n((K_1,K_2);\gamma,\lambda,\mu) -
\cE_{\max}((K_1,K_2))
$$
is non-increasing in $(-\pi,0]$ and non-decreasing in $[0,\pi]$.
\end{lemma}

\begin{proof}
Without loss of generality we assume that $i=1.$ Given ${\psi}\in
L^{2,e}(\T^2)$ consider
$$
(({H}_0(K) - \cE_{\min}(K)){\psi},{\psi})=\int_{\T^2}
\sum\limits_{i=1}^2 \cos\tfrac{K_i}{2}\,\big(1-\cos
q_i\big)|\psi(q)|^2\,\mathrm{d} q, \quad K:=(K_1,K_2).
$$
Clearly, the map $K_1\in\T\mapsto (({H}_0(K) -
\cE_{\min}(K)){\psi},{\psi})$ is non-decreasing in $(-\pi,0]$ and is
non-increasing in $[0,\pi].$ Since ${V}_{\gamma\lambda\mu}$ is
independent of $K,$ by definition of $e_n(K;\gamma,\lambda,\mu)$ the
map $K_1\in\T\mapsto e_n(K;\gamma,\lambda,\mu) - \cE_{\min}(K)$ is
non-decreasing in $(-\pi,0]$ and is non-increasing  in $[0,\pi].$

The case of $K_i\mapsto E_n(K;\gamma,\lambda,\mu) - \cE_{\max}(K)$
is similar.
\end{proof}

The following result represents an extension of \cite[Theorems 1 and
2]{ALMM:2006} and \cite[Theorem 3.1]{LKhKhamidov:2021}. It may be
considered as the first main result of the paper.

\begin{theorem}\label{teo:disc_Kvs0}
Assume that $\gamma, \lambda, \mu \in \mathbb{R}$. If the operator
$H_{\gamma\lambda\mu}(0)$ has $n_-$ eigenvalues below and $n_+$
eigenvalues above the essential spectrum, then for any $K \in \T^2$,
the fiber Hamiltonian $H_{\gamma \lambda \mu}(K)$ has at least $n_-$
eigenvalues below the essential spectrum and at least $n_+$ above
it. In each case the number $n_\pm$ is taken counting multiplicity
of the respective eigenvalues.
\end{theorem}
\begin{proof}
By  using Lemma \ref{lem:monoton_xos_qiymat},  for any $K\in\T^2$
and $m\ge1$ one derives the following inequalities
\begin{equation}\label{eq:min_eigen0}
0\le \cE_{\min}(0) - e_m(0;\gamma,\lambda,\mu) \le \cE_{\min}(K) -
e_m(K;\gamma,\lambda,\mu),
\end{equation}
and
$$
E_m(K;\gamma,\lambda,\mu) - \cE_{\max}(K) \ge
E_m(0;\gamma,\lambda,\mu) - \cE_{\max}(0) \ge 0.
$$
Assume that, counting multiplicities, the operator
$H_{\gamma\lambda\mu}(0)$ has $n\geq 1$ eigenvalues below the
essential spectrum, i.e. $ e_n(0;\gamma,\lambda,\mu)$ is a discrete
eigenvalue of ${H}_{\gamma\lambda\mu}(0)$  for some
$\gamma,\lambda,\mu\in\R.$ Then, $\cE_{\min}(0) -
e_n(0;\gamma,\lambda,\mu)>0,$ and hence, by \eqref{eq:min_eigen0}
and \eqref{eq:ess_spec} $e_n(K;\gamma,\lambda,\mu)$ is a discrete
eigenvalue of ${H}_{\gamma\lambda\mu}(K)$ for any $K\in\T^2.$ Since,
by definition, $e_1(K;\gamma,\lambda,\mu)\le \ldots \le
e_n(K;\gamma,\lambda,\mu)<\cE_{\min}(K),$ it follows that,  counting
multiplicities, ${H}_{\gamma\lambda\mu}(K)$ has at least $n$
eigenvalues below the essential spectrum. The case of
$E_n(K;\gamma,\lambda,\mu)$ is similar.
\end{proof}

Theorem \ref{teo:disc_Kvs0} implies that the numbers of discrete
eigenvalues of $H_{\gamma\lambda\mu}(0)$ below and above the
essential spectrum provide us with exact lower bounds (over all
$K\in \T^2$) for the corresponding numbers of discrete eigenvalues
of $H_{\gamma\lambda\mu}(K)$, $K\in\T^2$.

\subsection{Reducing subspaces of the fiber Schr\"odinger
operators $H_{\gamma\lambda\mu}(0)$}\label{sec:ReduSub}

In this subsection we show that for any
$\gamma,\lambda,\mu\in\mathbb{R}$ the subspaces
$L^{2,\mathrm{ees}}(\T^2)$,\,\,$L^{2,\mathrm{oos}}(\T^2)$ \, and \,
$L^\mathrm{2,ea}(\T^2)$  introduced in
\eqref{def:ees_space}--\eqref{def:ea_space}, form a complete set of
mutually orthogonal reducing subspaces of $H_{\gamma\lambda\mu}(0)$.

\begin{lemma}\label{lem:subspaces}
The followings statements hold true:
\begin{itemize}
\item[(i)]
The subspaces
$L^{2,\mathrm{ees}}(\T^2)$,\,\,$L^{2,\mathrm{oos}}(\T^2)$ \, and \,
$L^\mathrm{2,ea}(\T^2)$ are mutually orthogonal and $L^{2,\rm
e}(\T^2)=L^{2,\mathrm{ees}}(\T^2) \oplus L^{2,\mathrm{oos}}(\T^2)
\oplus L^\mathrm{2,ea}(\T^2)$;
\item[(ii)] The subspaces
$L^{2,\mathrm{ees}}(\T^2)$,\,\,$L^{2,\mathrm{oos}}(\T^2)$ \, and \,
$L^\mathrm{2,ea}(\T^2)$ are reducing for $H_{\gamma \lambda
\mu}(0)$.
\end{itemize}
\end{lemma}

\begin{proof}
\rm{(i)} Let $f\in L^{2,\rm e}(\T^2)$ and
\begin{align*}
&f_{\mathrm{ees}}(p_1,p_2)=\frac{f(p_1,p_2)+f(p_2,p_1)+f(p_1,-p_2)+f(-p_2,p_1)}{4},\\
& f_\mathrm{oos}(p_1,p_2)=\frac{f(p_1,p_2)+f(p_2,p_1)-f(p_1,-p_2)-f(-p_2,p_1)}{4},\\
&f_\mathrm{ea}(p_1,p_2)=\frac{f(p_1,p_2)-f(p_2,p_1)}{2}.
\end{align*}
Then  $f=f_{\mathrm{ees}}+f_\mathrm{oos}+f_\mathrm{ea}$ and
$f_{\mathrm{ees}}\in L^{2,\mathrm{\mathrm{ees}}}(\T^2),
f_\mathrm{oos}\in L^{2,\mathrm{oos}}(\T^2), f_\mathrm{ea}\in
L^{2,\mathrm{ea}}(\T^2)$.

On the other hand, for every  $h\in L^{2,\mathrm{ea}}(\T^2)$ the
followings equalities hold:
\begin{align*}
\int\limits_{\T^2} h(x,y)dxdy=&\int\limits_{\T}
\biggl(\int\limits_{\T} h(x,y)dx\biggr)dy=\int\limits_{\T}
\biggl(\int\limits_{\T} h(x,y)dy\biggr)dx\\
&= \int\limits_{\T} \biggl(\int\limits_{\T}
-h(y,x)dy\biggr)dx=-\int\limits_{\T^2} h(x,y)dxdy.
\end{align*}
Thus, $\int\limits_{\T^2} h(x,y)dxdy=0.$  By definitions
\eqref{def:ees_space}--\eqref{def:ea_space}, for any
$f_{\mathrm{ees}}\in L^{2,\mathrm{ees}}(\T^2)$,\, $f_\mathrm{oos}\in
L^{2,\mathrm{oos}}(\T^2)$ and $f_\mathrm{ea}\in
L^{2,\mathrm{ea}}(\T^2)$ the relations
$f_{\mathrm{ees}}f_\mathrm{ea}\in L^{2,\mathrm{ea}}(\T^2)$ and
$f_\mathrm{oos}f_\mathrm{ea}\in L^{2,\mathrm{ea}}(\T^2)$ hold true,
which implies that $f_{\mathrm{ees}}\perp f_\mathrm{ea}$ and $
f_\mathrm{oos}\perp f_\mathrm{ea}.$ Similarly, if $ h\in
L^{2,\mathrm{oos}}(\T^2)$ then
$$ \int\limits_{\T^2} h(x,y)dxdy=\int\limits_{\T} \biggl(\int\limits_{\T}
h(y,x)dy\biggr)dx= \int\limits_{\T} 0\, dx=0.
$$
Thus, $f_\mathrm{oos}f_{\mathrm{ees}}\in L^{2,\mathrm{oos}}(\T^2)$
and hence $f_\mathrm{oos}\perp f_{\mathrm{ees}}.$

\rm{(ii)} Recall that $H_0(0)$ is the multiplication operator by the
function $\cE_0(p)=2\epsilon(p)\in L^{2,\mathrm{ees}}(\T^2)$. This
yields that the invariant subspaces $L^{2,\mathrm{ees}}(\T^2) $,
$L^{2,\mathrm{oos}}(\T^2)$ and $ L^\mathrm{2,ea}(\T^2)$ are reducing
for $H_0(0)$.

Taking into account the identity
\begin{align*}
2\cos x \cos y &+ 2\cos z \cos t =(\cos x + \cos z)(\cos y + \cos t)
 + (\cos x - \cos
z)(\cos y - \cos t),\\
&\text{for any \,}x,y,z,t\in\R,
\end{align*}
from \eqref{moment_poten} one derives that
\begin{align}\label{moment_poten_2}
(V_{\gamma\lambda\mu}f)(p) &=\frac{\gamma }{4\pi^2}\int_{\T^2}
f(q)\,\d q  + \frac{\lambda }{8\pi^2}(\cos p_1 + \cos p_2)
\int_{\T^2} (\cos q_1 + \cos q_2)f(q)\,\d q\\ \nonumber &+
\frac{\lambda }{8\pi^2}(\cos p_1 - \cos p_2) \int_{\T^2} (\cos q_1 -
\cos q_2)f(q)\,\d q\\ \nonumber &+\frac{\mu }{8\pi^2}(\cos 2p_1 +
\cos 2p_2) \int_{\T^2} (\cos 2q_1 + \cos 2q_2)f(q)\,\d q\\ \nonumber
&+\frac{\mu }{8\pi^2}(\cos 2p_1 - \cos 2p_2) \int_{\T^2} (\cos 2q_1
- \cos 2q_2)f(q)\,\d q\\ \nonumber & +\frac{\mu }{2\pi^2}\cos
p_1\cos p_2 \int_{\T^2} \cos q_1\cos q_2f(q)\,\d q\\ \nonumber
&+\frac{\mu }{2\pi^2}\sin p_1\sin p_2 \int_{\T^2} \sin q_1\sin
q_2f(q)\,\d q,
\end{align}
which implies that the subspaces $L^{2,\mathrm{ees}}(\T^2) $,
$L^{2,\mathrm{oos}}(\T^2)$ and $ L^{2,\rm ea}(\T^2)$ are reducing
for $V_{\gamma\lambda\mu}$ and, hence, for
$H_{\gamma\lambda\mu}(0)$.
\end{proof}

\section{Main results: The discrete spectrum of the
two-boson fiber Hamiltonians}\label{sec:main_results}

From Lemma \ref{lem:subspaces}\,(ii) it immediately follows that
\begin{equation}\label{separation_K0}
\sigma\bigl(H_{\gamma \lambda \mu }(0)\bigr) =
 \bigcup_{\theta \in \{\mathrm{oos},\mathrm{ea},\mathrm{ees}\} }
\sigma \bigl(
H_{\gamma\lambda\mu}(0)\big|_{L^{2,\theta}(\T^2)}\bigr),
\end{equation}
where $A\big|_\cK$ stands for the restriction of a bounded  operator
$A$ on its reducing subspace~$\cK.$ Therefore, in order to study the
number of eigenvalues of the total operator
${H}_{\gamma\lambda\mu}(0)$ in the space $L^{2,\rm e}(\T^2)$, it
suffices to separately evaluate the numbers of eigenvalues of the
parts of ${H}_{\gamma\lambda\mu}(0)$ in the reducing subspaces
$L^{2,\mathrm{ees}}(\T^2)$, $L^{2,\mathrm{oos}}(\T^2)$, and
$L^{2,\mathrm{ea}}(\T^2)$. Below, we will use the following
notations
\begin{align}
\label{Hoos}
{H}_{\mu}^\mathrm{oos}(0)&:={H}_{\gamma\lambda\mu}(0)\big|_{L^{2,\mathrm{oos}}(\T^2)}
= H_0(0) + {V}_{\mu}^\mathrm{oos},\\
\label{Hea}
{H}_{\lambda\mu}^\mathrm{ea}(0)&:={H}_{\gamma\lambda\mu}(0)\big|_{L^{2,\mathrm{ea}}(\T^2)}
= H_0(0) + {V}_{\lambda\mu}^\mathrm{ea},\\
\label{Hees}
{H}_{\gamma\lambda\mu}^\mathrm{ees}(0)&:={H}_{\gamma\lambda\mu}(0)
\big|_{L^{2,\mathrm{ees}}(\T^2)}=H_0(0) +
{V}_{\gamma\lambda\mu}^\mathrm{ees}
\end{align}
for the respective parts of the Schr\"odinger operator
$H_{\gamma\lambda\mu}(0)$, $\gamma,\lambda,\mu\in\mathbb{R}$, in its
reducing subspaces
$L^{2,\mathrm{ees}}(\T^2)$,\,\,$L^{2,\mathrm{oos}}(\T^2)$ \, and \,
$L^\mathrm{2,ea}(\T^2)$. Here, by ${V}_{\mu}^\mathrm{oos}$,
${V}_{\lambda\mu}^\mathrm{ea}$ and
${V}_{\gamma\lambda\mu}^\mathrm{ees}$ we denote the corresponding
parts of the potential ${V}_{\gamma\lambda\mu}$. From
\eqref{moment_poten_2} it follows that
   \begin{align}
\label{Voos}
 (V_{\mu }^\mathrm{oos}f)(p)&=\frac{\mu }{2\pi^2}\sin p_1\sin p_2
\int_{\T^2} \sin q_1\sin q_2f(q)\,\d q,\\
\label{Vea} (V_{\lambda \mu }^\mathrm{ea}f)(p) &=\frac{\lambda
}{8\pi^2}(\cos
p_1 - \cos p_2) \int_{\T^2} (\cos q_1 - \cos q_2)f(q)\,\d q \\
\nonumber &+ \frac{\mu }{8\pi^2}(\cos 2p_1 - \cos 2p_2) \int_{\T^2}
(\cos 2q_1 - \cos 2q_2)f(q)\,\d q
\end{align}
and
\begin{align}
\label{Vees} (V_{\gamma \lambda \mu }^\mathrm{ees}f)(p) &=
\frac{\gamma }{4\pi^2}\int_{\T^2} f(q)\,\d q  + \frac{\lambda
}{8\pi^2}(\cos p_1 + \cos p_2) \int_{\T^2} (\cos q_1 + \cos
q_2)f(q)\,\d q\\ \nonumber &+\frac{\mu }{8\pi^2}(\cos 2p_1 + \cos
2p_2) \int_{\T^2} (\cos 2q_1 + \cos 2q_2)f(q)\,\d q\\ \nonumber
&+\frac{\mu }{2\pi^2}\cos p_1\cos p_2 \int_{\T^2} \cos q_1\cos
q_2f(q)\,\d q.
\end{align}

\subsection{The discrete spectrum of ${H}_{\mu}^\mathrm{oos}(0)$.}
We begin consideration of the discrete spectrum of the operator
$H_{\gamma\lambda\mu}(0)$ with the study of its simplest part
$H^{\rm oos}_{\mu}(0)$ defined in \eqref{Hoos}, \eqref{Voos}. Let us
introduce the following sets:
\begin{align*}
\cS^{+}_{0}=(-\infty,\tfrac{3\pi}{3\pi-8}], \quad
\cS^{+}_{1}=(\tfrac{3\pi}{3\pi-8},+\infty)
\end{align*}
and
\begin{align*}
\cS^{-}_{0}=[-\tfrac{3\pi}{3\pi-8},+\infty), \quad
\cS^{-}_{1}=(-\infty,-\tfrac{3\pi}{3\pi-8})
\end{align*}
\begin{theorem}\label{teo:eigs_of_H_oos}
Let $\mu\in\mathbb{R}$ and $\alpha=0,1$. If  $\mu\in \cS_\alpha^+$
resp.  $\mu\in \cS_\alpha^-,$  then $H^{\rm oos}_{\mu}(0)$ has
exactly $\alpha$  eigenvalues above resp. below the essential
spectrum.
\end{theorem}

\begin{proof}
In the case under consideration, the perturbation operator,
$V_\mu^{\rm oos}$, is given by \eqref{Voos}. Clearly, it is of rank
one. In such a case the proof would literally repeat the proof of an
analogous result in \cite[Lemma 4.1]{LKhKhamidov:2021}.  Thus, we
skip it.
\end{proof}

\subsection{The discrete spectrum of ${H}_{\lambda\mu}^\mathrm{ea}(0)$.}

Recall that the part ${H}_{\lambda\mu}^\mathrm{ea}(0)$ of the fiber
Schr\"o\-din\-ger operator ${H}_{\lambda\mu}(0)$ in the reducing
subspace $L^{2,\mathrm{ea}}(\T^2)$ has the form \eqref{Hea} with the
potential $V_{\lambda \mu }^\mathrm{ea}$ given by \eqref{Vea}.
Spectral properties of this operator have already been studied in
\cite{LHA:2024}. In order to describe the results of \cite{LHA:2024}
concerning the number and location of eigenvalues with
(antisymmetric) eigenfunctions of the Hamiltonian $H^{\rm
ea}_{\lambda\mu}(0)$ for various $\lambda,\mu\in\mathbb{R}$, we
define
\begin{align}\label{root_0}
\mu^*:&=\tfrac{4(4-\pi)}{32-9\pi}
\end{align}
and introduce the following sets (see Fig. \ref{fig:sohalar_umumiy})
and introduce the following sets 
\begin{align}
\nonumber
\cA_{0}^\pm:=
&\Big\{(\lambda,\mu)\in\R^2:\,\,\mp\lambda\geq\frac{4}{\mp\mu +
\mu^*} + 8, \,\, \pm\mu<\mu^* \Big\},\\ \nonumber 
\cA_{1}^\pm:=
&\Big\{(\lambda,\mu)\in\R^2:\,\,\mp\lambda<\frac{4}{\mp\mu+\mu^*}
+ 8, \,\, \pm\mu<\mu^* \Big\} \\
\label{six_sets_p}
&\cup \Big\{(\lambda,\mu^*):\,\, \lambda\in\R
\Big\}\\ \nonumber
&\cup\Big\{(\lambda,\mu)\in\R^2:\,\,\mp\lambda\geq\frac{4}{\mp\mu
+ \mu^*} + 8, \,\,\,\, \pm\mu>\mu^* \Big\},\\ \nonumber 
\cA_{2}^\pm :=
&\Big\{(\lambda,\mu)\in\R^2:\,\,\mp\lambda<\frac{4}{\mp\mu +
\mu^*} + 8, \,\,\,\, \pm\mu>\mu^* \Big\}.
\end{align}

Obviously, both the "$+$" and "$-$" variants of the disjoint sets
$\cA_{j}^{+}$ and $\cA_{j}^{-}$, $j=0,1,2$, provide us with a
partition of the $(\lambda, \mu)$ plane $\R^2$ such that
$\cup_{\alpha = 0}^2 \cA_{j}^{\pm} = \R^2$.
\begin{figure}[h]
\includegraphics[width=0.65\textwidth]{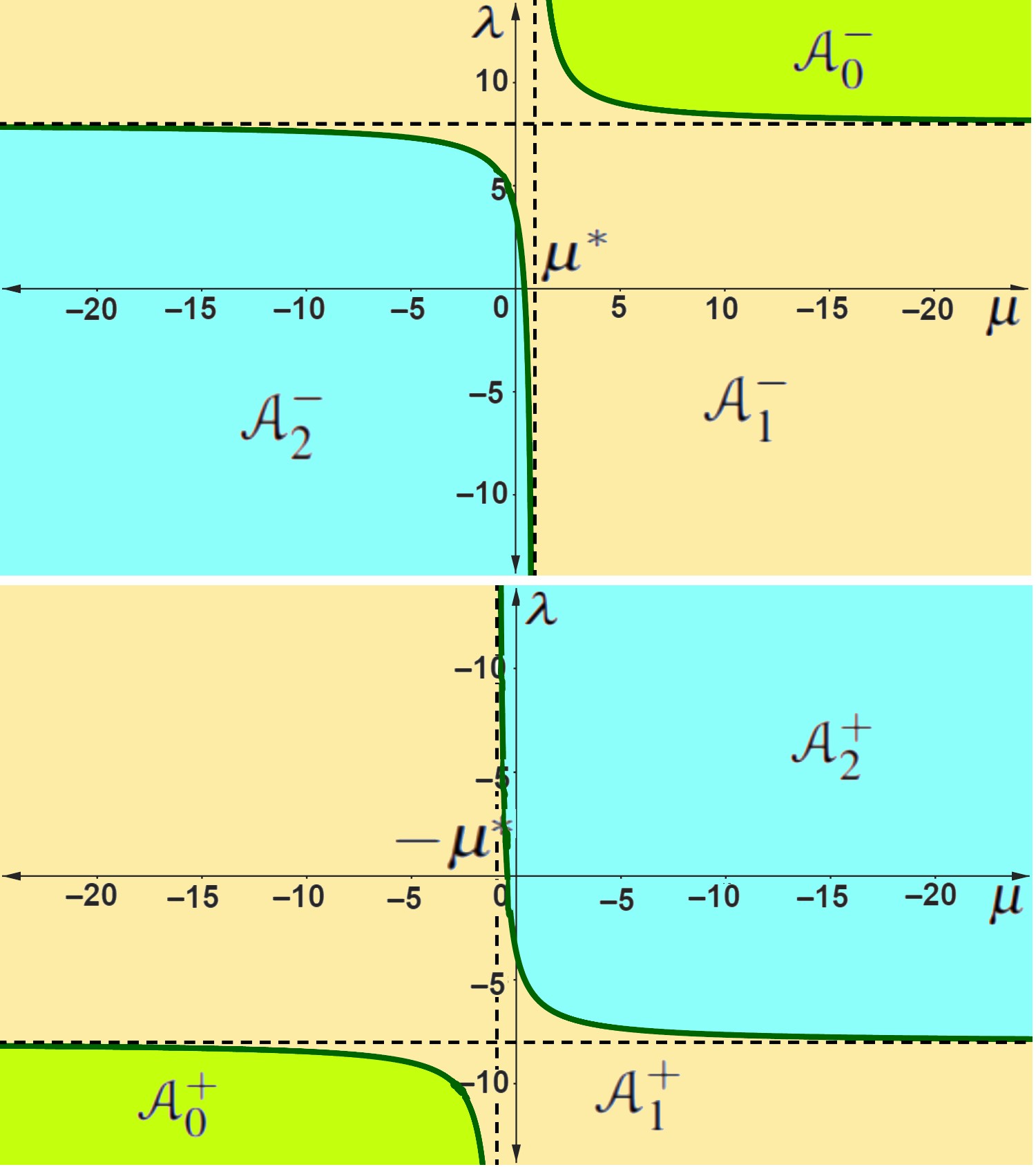}
\caption{A schematic location of the sets $\cA_{j}^\pm$, $j=0,1,2,$
defined in \eqref{six_sets_p}.}
\label{fig:sohalar_umumiy}
\end{figure}
\begin{theorem}[\cite{LHA:2024}, Theorem 3.3]\label{teo:eigens_of_H_ea}
Let $\lambda,\mu\in\mathbb{R}$ and $\beta=0,1,2$. If
$(\lambda,\mu)\in \cA_{\beta}^+$ \, (resp.  $(\lambda,\mu)\in
\cA_{\beta}^-$ \,) then the operator $H^\mathrm{ea}_{\lambda\mu}(0)$
has exactly $\beta$ eigenvalues above (resp. below) its essential
spectrum.
\end{theorem}

Notice that the proof of this statement in \cite{LHA:2024} is quite
similar to the proof of the analogous results for the case of
rank-two perturbations established in \cite[Theorem
3.4]{LKhKhamidov:2021} and \cite[Theorems 3.2 and
3.4]{LMAbduhakimov:2023}.

\subsection{The discrete spectrum of ${H}_{\gamma\lambda\mu}^\mathrm{ees}(0)$.}

The part ${H}_{\gamma\lambda\mu}^\mathrm{ees}(0)$ of the fiber
Schr\"o\-din\-ger operator ${H}_{\gamma\lambda\mu}(0)$ in the
reducing subspace $L^{2,\mathrm{ees}}(\T^2)$ is given by \eqref{Hea}
where the perturbation $V_{\gamma\lambda \mu }^\mathrm{ees}$ is of
the form \eqref{Vees}. Since  $V_{\gamma\lambda \mu }^\mathrm{ees}$
depends on the three independent parameters and its rank is four,
the consideration of ${H}_{\gamma\lambda\mu}^\mathrm{ees}(0)$ is
essentially more cumbersome than that of ${H}_{\mu}^\mathrm{oos}(0)$
and $H^{\rm ea}_{\lambda\mu}(0)$. By this reason, we transfer the
lengthiest proofs of the results in this subsection to Section
\ref{sec:proofs}.

We start with introducing the following four real numbers:
\begin{align}\label{roots0}
&\mu_0^{\pm}:= \tfrac{-80+24\pi \pm \sqrt{666\pi^2-4128\pi+6400}}{3(16-5\pi)},\\
\label{roots1} &\mu_1^{\pm}:= \tfrac{-80+21\pi \pm
\sqrt{1161\pi^2-5664\pi+6400}}{12(16-5\pi)}.
\end{align}
Notice that the numerical values of  $ \mu^{\pm}_{0}$ and
$\mu^{\pm}_{1}$ are as follows:
\begin{align*}
\mu_0^-=-7.7170..., \quad  \mu_0^+=-2.7880...,\quad
\mu_1^-=-7.2646...,\quad \mu_1^+=-0,7404...,
\end{align*}
and, hence, these numbers satisfy the relations
\begin{align}\label{ineq:roots}
&\mu^{-}_{0}<\mu_{1}^{-} <\mu^{+}_{0}<\mu^{+}_{1}<0.
\end{align}

Let $Q^\pm(\gamma,\lambda,\mu)$ be defined as
\begin{align*}
Q^\pm(\gamma,\lambda,\mu)=&\mp(\gamma +2\lambda +4\mu) +
\frac{1}{2}\gamma \lambda +\frac{39\pi-104}{3\pi}\gamma \mu +
\frac{153\pi-448}{6\pi}\lambda \mu +\\
&+\frac{15\pi-40}{3\pi}\mu^2 \mp(\frac{18\pi-52}{3\pi}\gamma \lambda
\mu
+\frac{2(900\pi-135\pi^2-1472)}{9\pi^2}\gamma \mu^2  +\\
&\frac{3720\pi-585\pi^2-5888}{9\pi^2}\lambda \mu^2)
+\frac{3720\pi-585\pi^2-5888}{32\pi^2}\gamma \lambda \mu^2.
\end{align*}
Then
\begin{equation}\label{polynomial:Q}
Q^{\pm}(\gamma, \lambda, \mu)=(\gamma \mp 4)\Big(\lambda
Q^{\pm}_{0}(\mu) \mp Q^{\pm}_1(\mu)\Big)-8Q^{\pm}_{0}(\mu),
\end{equation}
where
\begin{align}
\label{polynomial:q0}
 Q_0^{\pm}(\mu):= &\tfrac{16-5\pi}{4\pi}(\mu \pm \mu_{0}^{-})(\mu
\pm \mu_{0}^{+}),\\
\label{polynomial:q1} Q^{\pm}_1(\mu):= &\tfrac{2(16-5\pi)}{\pi}(\mu
\pm \mu_{1}^{-})(\mu \pm \mu_{1}^{+}).
\end{align}

We also introduce the following notations:
\begin{align}\label{def:interval_I}
&I^{-}_{1}=(\mu_{0}^{+},+\infty), \,\,
I^{-}_{2}=(\mu_{0}^{-},\mu_{0}^{+}) , \,\,
I^{-}_{3}=(-\infty,\mu_{0}^{-}); \\ \nonumber
&I^{+}_{1}=(-\infty,-\mu_{0}^{+}), \,\,
I^{+}_{2}=(-\mu_{0}^{+},-\mu_{0}^{-}) , \,\,
I^{+}_{3}=(-\mu_{0}^{-},+\infty).
\end{align}

\begin{lemma}\label{lem:regions}
The set of points $(\lambda,\mu)\in \R^2$ satisfying equation
$\lambda Q^{\pm}_{0}(\mu)\mp Q^{\pm}_1(\mu)= 0$ coincides  with the
graph of the function
\begin{equation}\label{def:function_lambda1}
\lambda^{\pm}(\mu)=\pm
\frac{Q_{1}^{\pm}(\mu)}{Q_{0}^{\pm}(\mu)},\quad
\mu\in\R\setminus\{\mp\mu_0^-,\mp\mu_0^+\}.
\end{equation}
This set consists of three isolated smooth unbounded connected
curves
\begin{align}\label{curves}
&\tau^{\pm}_{j}:=\left\{(\lambda,\mu)\in\R^2:\,\,
\lambda=\pm\frac{Q^{\pm}_{1}(\mu)}{Q^{\pm}_{0}(\mu)},\,\,\mu\in
I^{\pm}_{j}\right\},\,\, j=1,2,3,
\end{align}
and partitions the $(\lambda,\mu)$-plane into four unbounded
connected components
\begin{equation}
\label{def:regions_D}
\begin{array}{cl}
 D^{\pm}_1:=&\{(\lambda,\mu)\in\R^2:\,\,
\lambda>\lambda^{\pm}(\mu),\,\, \mu \in I^{\pm}_{1}\},\\
D^{\pm}_2: =&\{(\lambda,\mu)\in\R^2:\,\,
\lambda<\lambda^{\pm}(\mu),\,\, \mu \in I^{\pm}_{1}\}\cup
\{(\lambda,\mu)\in\R^2:\,\,
\mu=\mp \mu_0^{+} \} \\
&\{(\lambda,\mu)\in\R^2:\,\,
\lambda>\lambda^{\pm}(\mu),\,\, \mu \in I^{\pm}_{2}\},\\
D^{\pm}_3: =&\{(\lambda,\mu)\in\R^2:\,\,
\lambda<\lambda^{\pm}(\mu),\,\, \mu \in I^{\pm}_{2}\}\cup
\{(\lambda,\mu)\in\R^2:\,\,
\mu=\mp \mu_0^{-} \} \\
&\{(\lambda,\mu)\in\R^2:\,\,
\lambda>\lambda^{\pm}(\mu),\,\, \mu \in I^{\pm}_{3}\},\\
D^{\pm}_4:=&\{(\lambda,\mu)\in\R^2:\,\,
\lambda<\lambda^{\pm}(\mu),\,\,\mu \in I_{3}^{\pm}\}
\end{array}
\end{equation}
\end{lemma}

\begin{proof}
We notice that the system of equations
\begin{align}\label{system 10}
\begin{cases}
&\lambda Q^{\pm}_{0}(\mu)\mp Q^{\pm}_1(\mu)= 0\\
&Q^{\pm}_{0}(\mu)=0
\end{cases}
\end{align}
is equivalent to the system of equations  $Q^{\pm}_{0}(\mu)=
Q_1^{\pm}(\mu)= 0$ which has no solutions, since by
\eqref{ineq:roots} the roots  $ \mu_{0}^{-},\,\mu_{0}^{+}$ and  $
\mu_{1}^{-}, \,\mu_{1}^{+} $ of the polynomials $Q^{\pm}_{0}$ and
$Q^{\pm}_{1}$ are different. Therefore $Q^{\pm}_{0}(\mu)\lambda \mp
Q_1^{\pm}(\mu)= 0$ implies the inequality $Q^{\pm}_{0}(\mu) \neq 0$
and, consequently,  the
 function \eqref{def:function_lambda1}

is well defined on $\R\setminus\{\mp\mu_0^-,\mp\mu_0^+\}$. Thus, the
set of  points  $(\lambda,\mu)\in\mathbb{R}^2$ defined by the
equation $Q^{\pm}_{0}(\mu)\lambda \mp Q^{\pm}_1(\mu)= 0$ coincides
with the graph

\begin{equation}\label{graph}
\bigl\{(\lambda^{\pm}(\mu),\mu),\mu\in\R\setminus
\{\mp\mu_0^-,\mp\mu_0^+\}\bigr\}\subset\mathbb{R}^2
\end{equation}
of the function $\lambda^{\pm}(\mu)$.  Since the function
$\lambda^\pm$ is defined on three disjoint intervals of the real
line, its graph is composed of three smooth curves on the
two-dimensional plane $\R^2$, namely, of the curves   $\tau_j^\pm$,
$j=1,2,3$, defined in \eqref{curves}.  These curves are isolated,
unbounded, simple, smooth and separate the $(\lambda,\mu)$-plane
into four unbounded distinct regions  (see Figure
\ref{Fig_D_minus}). These completes the proof.
\end{proof}

\begin{figure}
\centering
\includegraphics[width=0.9\textwidth]{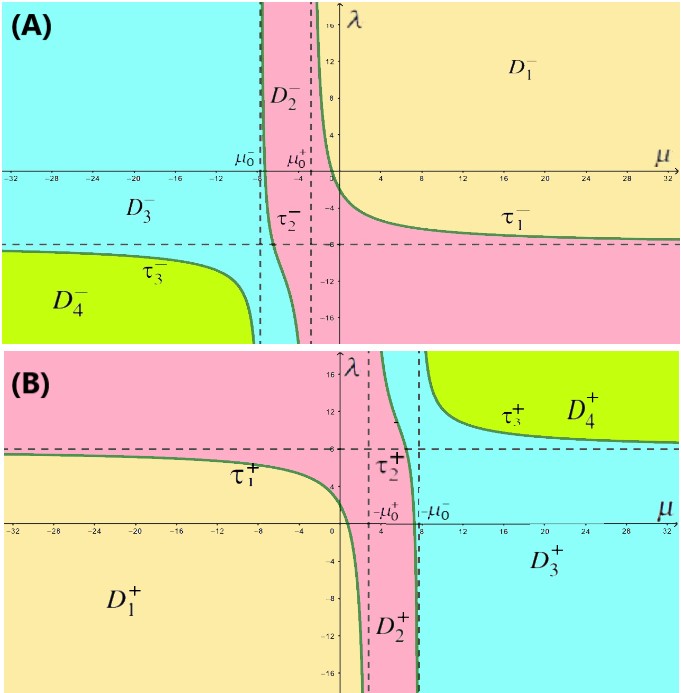}
\caption{ (A) Plots of the curves $\tau^-_j$, $j=1,2,3$,
partitioning the domain of the function  $\gamma^{-}$ into the parts
$D^-_\alpha$, $\alpha=1,2,3,4$. (B) Plots of the curves $\tau^+_j$,
$j=1,2,3$, partitioning the domain of the function  $\gamma^{+}$
into the parts $D^+_\alpha$, $\alpha=1,2,3,4$.} \label{Fig_D_minus}
\end{figure}

The next assertion describes relations between the sets
$D^{\pm}_{\alpha}, \alpha=1,2,3,4$ defined in Lemma
\ref{lem:regions} by formulas \eqref{def:regions_D}.

\begin{lemma}\label{lem:regions_D}
The following relations hold:
\begin{itemize}
\item[(i)]$D^{-}_4 \subseteq D^{+}_1$,
\item[(ii)]$D^{+}_4 \subseteq D^{-}_1$,
\item[(iii)]$D^{-}_4 \cap D^{+}_4=\emptyset$,
\item[(iv)] For any $\alpha=1,2,3,4$ the regions $D^{-}_\alpha$
and  $ D^{+}_\alpha$ are symmetric with respect to  the origin.
\end{itemize}
\end{lemma}

\begin{proof}
Statements (i)--(iii) follow immediately from the definitions
\eqref{def:regions_D} of the sets  $D^{\pm}_{\alpha},
\alpha=1,2,3,4$. Item (iv) is a consequence of the identity
$Q^{+}(\gamma,\lambda,\mu)=Q^{-}(-\gamma,-\lambda,-\mu)$ for any
$\gamma,\lambda,\mu\in\R$ (also see Figure \ref{Fig_D_minus}).
\end{proof}

Proof of the next statement is slightly more complicated but its
basic idea is exactly the same as that of Lemma~\ref{lem:regions}.

\begin{lemma}\label{lem:function_gamma}
The set of points $(\gamma,\lambda,\mu)\in\mathbb{R}^3$ satisfying
the equation $Q^{\pm}(\gamma,\lambda,\mu)= 0$ coincides with the
graph of the function
\begin{equation}\label{functiongamma}
\gamma^{\pm}(\lambda,\mu)=\frac{8Q^{\pm}_{0}(\mu)}{ \lambda
Q^{\pm}_{0}(\mu) \mp Q^{\pm}_1(\mu)} \pm 4, \quad
(\lambda,\mu)\in\R^2\setminus\bigcup\limits_{i=1}^3 \tau_i^\pm,
\end{equation}
where $Q^{\pm}_0(\cdot)$ and $Q^{\pm}_{1}(\cdot)$ are defined by
\eqref{polynomial:q0} and \eqref{polynomial:q1}, respectively, and
$\tau^\pm_i$, $i=1,2,3,$ are the null sets \eqref{curves} of the
denominators $\lambda Q^{\pm}_{0}(\mu) \mp Q^{\pm}_1(\mu)$. The
graph  of the function \eqref{functiongamma} consists of four
isolated unbounded smooth  connected surfaces
\begin{align}\label{def:Gamma}
&\Gamma^{\pm}_{j}=\{(\gamma,\lambda,\mu)\in\R^3:\gamma=\gamma^\pm(\lambda,\mu),\,\,
(\lambda,\mu)\in D^{\pm}_j\},\quad j=1,2,3,4,
\end{align}
that  separate
 the $(\gamma,\lambda,\mu)$-space into five unbounded
 connected components
\begin{equation}
\label{def:ees_sets}
\begin{array}{cl}
\cC^{\pm}_{0}=&\{(\gamma,\lambda,\mu)\in\R^3:\gamma>
\gamma^\pm(\lambda,\mu), \,\,(\lambda,\mu)\in D^{\pm}_1\},\\
\cC^{\pm}_{1}=&\{(\gamma,\lambda,\mu)\in\R^3:
\gamma<\gamma^\pm(\lambda,\mu),\,\, (\lambda,\mu)\in D^{\pm}_1\}\cup \\
&\{(\gamma,\lambda,\mu)\in\R^3: (\lambda,\mu)\in \tau^{\pm}_1\}\cup \\
&\{(\gamma,\lambda,\mu)\in\R^3:
\gamma>\gamma^\pm(\lambda,\mu),\,\, (\lambda,\mu)\in D^{\pm}_2 \},\\
\cC^{\pm}_{2}=&\{(\gamma,\lambda,\mu)\in\R^3:
\gamma<\gamma^\pm(\lambda,\mu),\,\, (\lambda,\mu)\in D^{\pm}_2\}\cup \\
&\{(\gamma,\lambda,\mu)\in\R^3: (\lambda,\mu)\in \tau^{\pm}_2\}\cup \\
&\{(\gamma,\lambda,\mu)\in\R^3:
\gamma>\gamma^\pm(\lambda,\mu),\,\, (\lambda,\mu)\in D^{\pm}_3 \},\\
\cC^{\pm}_{3}=&\{(\gamma,\lambda,\mu)\in\R^3:
\gamma<\gamma^\pm(\lambda,\mu),\,\, (\lambda,\mu)\in D^{\pm}_3\}\cup \\
&\{(\gamma,\lambda,\mu)\in\R^3: (\lambda,\mu)\in \tau^{\pm}_3\}\cup \\
&\{(\gamma,\lambda,\mu)\in\R^3:
\gamma>\gamma^\pm(\lambda,\mu),\,\, (\lambda,\mu)\in D^{\pm}_4 \},\\
\cC^{\pm}_{4}=&\{(\gamma,\lambda,\mu)\in\R^3:
\gamma<\gamma^\pm(\lambda,\mu), (\lambda,\mu)\in  D^{\pm}_4\}.
\end{array}
\end{equation}
\end{lemma}
\begin{proof}
First, one observes that the system of equations
\begin{align}\label{system1}
\begin{cases}
&Q^{\pm}(\gamma, \lambda, \mu)=0 \\
&\lambda Q^{\pm}_{0}(\mu) \mp Q_1^{\pm}(\mu)=0
\end{cases}
\end{align}
has no solutions. Indeed, if \eqref{system1} has a solution then
\eqref{polynomial:Q} immediately implies
\begin{equation}\label{different}
Q^{\pm}_{0}(\mu)=Q^{\pm}_{1}(\mu)=0.
\end{equation}
However this contradicts the fact that the roots  $
\mp\mu_{0}^{-},\,\mp\mu_{0}^{+}$ and  $\mp\mu_{1}^{-},
\,\mp\mu_{1}^{+} $ of the polynomials $Q^{\pm}_{0}$ and
$Q^{\pm}_{1}$ are all different \eqref{ineq:roots}.

Taking into account that the system \eqref{system1} has no
solutions, the equality $Q^{\pm}(\gamma,\lambda,\mu)=0$ implies the
inequality $\lambda Q^{\pm}_{0}(\mu) \mp Q_1^{\pm}(\mu) \neq 0$
which means that the function
\begin{equation}\label{functiongamma^{-}}
\gamma^{\pm}(\lambda,\mu)=\frac{8Q^{\pm}_{0}(\mu)}{ \lambda
Q^{\pm}_{0}(\mu) \mp Q^{\pm}_1(\mu)}\pm 4
\end{equation}
is well defined on $\R^2\setminus\bigcup\limits_{i=1}^3 \tau_i^\pm$.
Now, let $(\gamma,\lambda,\mu)\in\mathbb{R}^3$ belong to the graph
of the function $\gamma^{\pm}(\lambda,\mu)$, i.e.,
\begin{equation}\label{functiongamma1}
\gamma=\frac{8Q^{\pm}_{0}(\mu)}{\lambda Q^{\pm}_{0}(\mu) \mp
Q^{\pm}_1(\mu)}\pm 4.
\end{equation}
Combining \eqref{functiongamma1} with \eqref{polynomial:Q} yields
the first equality of \eqref{system1}. Thus, we have showed that the
set of the points  $(\gamma,\lambda,\mu)\in\mathbb{R}^3$ defined by
the equation $Q^{\pm}(\gamma,\lambda,\mu)= 0$ coincides with the
graph $(\gamma^{\pm}(\lambda,\mu),\lambda,\mu)\in\mathbb{R}^3$ of
the smooth function $\gamma^{\pm}$.

Similarly, in Lemma \ref{lem:regions}, the domain of the function
$\gamma^\pm$ consists of four disjoint unbounded regions
$D^\pm_j\subset\mathbb{R}^2$, $j=1,2,3,4,$ defined in
\eqref{def:regions_D}. This implies that the graph of the function
$\gamma^\pm$ consists of four separate unbounded simple smooth
surfaces $\Gamma_j^\pm$, $j = 1, 2, 3, 4$, given by
\eqref{def:Gamma}. These surfaces partition the space $\mathbb{R}^3$
into five separate unbounded connected components, as defined in
\eqref{def:ees_sets}. This completes the proof.
\end{proof}

\begin{remark}\label{remark:Gamma_surfaces}
Recall that the domain of the function $\gamma^{\pm}(\cdot,\cdot)$
is the set
\begin{align*}
&\R^2\setminus \bigcup\limits_{i=1}^3  \tau^{\pm}_i
=\bigcup\limits_{j=1}^4 D^{\pm}_{j}.
\end{align*}
The curve $\tau^{\pm}_j$ in $\R^2$ defines the corresponding
cylinder
\begin{align*}
&\Upsilon^{\pm}_j:=\bigl\{\bigl(\gamma,\lambda,\mu\bigr)\in\R^3\big|\,
(\lambda,\mu)\in \tau^{\pm}_j\bigr\}, \quad j=1,2,3,
\end{align*}
that represents an asymptotic surface for the graph of $\gamma^\pm$.
The cylinders $\Upsilon^{\pm}_1$,$\Upsilon^{\pm}_2$  and
$\Upsilon^{\pm}_3$ separate the graph of the function
$\gamma^{\pm}(\cdot,\cdot)$    into  the four isolated smooth
connected unbounded surfaces $\Gamma^{\pm}_{1}$, $\Gamma^{\pm}_{2}$,
$\Gamma^{\pm}_{3}$ and $\Gamma^{\pm}_{4}$ in $\R^3$, respectively.

\end{remark}
\begin{remark}
\label{Rem_Q_sigdef} By construction of the domains
\eqref{def:ees_sets}, the  function $Q^\pm(\gamma,\lambda,\mu)$
keeps it sign on every respective domain $\cC^\pm_\alpha$,
$\alpha=0,1,2,3,4$.
\end{remark}

Now we evaluate the number of bound states of the operator
$H_{\gamma \lambda \mu }^\mathrm{ees}(0)$ with energies lying below
the essential spectrum and above it.

\begin{theorem}\label{teo:constant}
Let $\cC^+$ be one of the open connected components $\cC^+_\zeta$,
$\zeta=0,1,2,3,4$, of the $(\gamma,\lambda,\mu)$-space defined in
\eqref{def:ees_sets} with the sign ``$+$''. Then for any
$(\gamma,\lambda,\mu)\in\cC^+$ the number
$n_+({H}^{\mathrm{ees}}_{\gamma\lambda\mu}(0))$ of eigenvalues of
$H^{\mathrm{ees}}_{\gamma\lambda\mu}(0)$ lying above the essential
spectrum
$\sigma_\mathrm{ess}\bigl(H^{\mathrm{ees}}_{\gamma\lambda\mu}(0)\bigr)$
remains constant.

Analogously, let $\cC^-$ be one of the above open connected
components $\cC^-_\zeta$, $\zeta=0,1,2,3,4$,  of the
$(\gamma,\lambda,\mu)$-space defined in \eqref{def:ees_sets} with
the sign ``$-$''. Then for any $(\gamma,\lambda,\mu)\in\cC^-$ the
number $n_-({H}^\mathrm{ees}_{\gamma\lambda\mu}(0))$ of eigenvalues
of $H^\mathrm{ees}_{\gamma\lambda\mu}(0)$ lying below the essential
spectrum
$\sigma_\mathrm{ess}\bigl(H_{\gamma\lambda\mu}^\mathrm{ees}(0)\bigr)$
remains constant.
\end{theorem}

We postpone the proof of this results as well as the proof of
Theorem \ref{teo:ee,s}  till Sec. \ref{sec:proofs}.

\begin{theorem}\label{teo:ee,s}
Let $\gamma,\lambda,\mu\in \mathbb{R}$  and $ \zeta=0,1,2,3,4$. If
$( \gamma,\lambda,\mu)\in \mathcal{C}^{-}_{\zeta} $ (resp.  $(
\gamma,\lambda,\mu)\in \mathcal{C}^{+}_{\zeta} $)  then
$H^\mathrm{ees}_{ \gamma\lambda\mu}(0)$ has exactly $\zeta$ bound
states with energies lying  below  (resp. above) the essential
spectrum.
\end{theorem}

Let
$$
\mathcal{P}^{\pm}_{\zeta}:=\mathcal{C}^{\pm}_{\zeta}\cup
\Gamma^{\pm}_{\zeta+1} \quad\text{if \, $\zeta=0,1,2,3$}\quad
\text{and \, }\mathcal{P}^{\pm}_{4}:=\mathcal{C}^{\pm}_{4}.
$$
By construction, the set $\{
\mathcal{P}^{\pm}_{0},\mathcal{P}^{\pm}_{1},\mathcal{P}^{\pm}_{2},
\mathcal{P}^{\pm}_{3},\mathcal{P}^{\pm}_{4} \}$ is a partition of
the space $\R^3$.
\begin{corollary}\label{corr:eig_of_H(0)}
The number of eigenvalues of the operator
$H^{\mathrm{ees}}_{\gamma\lambda\mu}(0),$ that lie above (resp.
below) the essential spectrum remains constant in each set
$\mathcal{P}^+\in \{
\mathcal{P}^{+}_{0},\mathcal{P}^{+}_{1},\mathcal{P}^{+}_{2},
\mathcal{P}^{+}_{3},\mathcal{P}^{+}_{4} \}$ (resp. in each set
$\mathcal{P}^-\in \{
\mathcal{P}^{-}_{0},\mathcal{P}^{-}_{1},\mathcal{P}^{-}_{2},
\mathcal{P}^{-}_{3},\mathcal{P}^{-}_{4} \}$).
\end{corollary}
\begin{proof}
The function $n_{+}(H^{\mathrm{ees}}_{\gamma\lambda\mu}(0)$ (resp.
$n_{-}(H^{\mathrm{ees}}_{\gamma\lambda\mu}(0)$)  is  continuous  on
each $ \mathcal{P}^{+}_{\zeta}$ (resp. on each
$\mathcal{P}^{-}_{\zeta}),$ $\zeta=0,1,2,3$). This proves the
assertion.
\end{proof}

\subsection{The total discrete spectrum
of $H_{\gamma\lambda\mu}(K)$} Introduce the following sets:
$$
\cG^{\pm}_{n}:=\bigcup_{(\alpha,\beta,\zeta)\in\cM_n}\,
\cS^\pm_{\alpha}\cap \cA^{\pm}_{\beta} \cap \cC^{\pm}_{\zeta}, \quad
n\in\N\cup\{0\}, \quad n\leq 7,
$$
where
\begin{align*}
\cM_n=&\bigl\{(\alpha,\beta,\zeta)\,\big|\, \alpha=0,1,\,
\beta=0,1,2, \,\, \zeta=0,1,2,3,4,\,\, \alpha+\beta+\zeta=n \bigr\}.
\end{align*}

Let
\begin{equation}\label{def:G_mn}
\mathbb{G}_{mn}=\cG^{-}_{m}\cap \cG^{+}_{n}, \quad
m,n=0,1,2,3,4,5,6,7.
\end{equation}
In view of \eqref{separation_K0}, Theorems \ref{teo:eigs_of_H_oos},
\ref{teo:eigens_of_H_ea} and \ref{teo:ee,s} lead to the following
conclusion.

\begin{corollary}\label{teo:eig_of_H(0)}
Let $\gamma,\lambda,\mu\in \mathbb{R}$  and $ m, n=0,1,2,3,4,5,6,7$.
If  $m+n\leq 7$ and  $( \gamma,\lambda,\mu)\in \mathbb{G}_{mn} $
then
$$
n_{-}\left(H_{ \gamma\lambda\mu}(0)\right)= m \quad \text{and} \quad
n_{+}\left( H_{ \gamma\lambda\mu}(0)\right)= n.
$$
\end{corollary}

The next theorem  gives estimates for the respective numbers
$n_\pm({H}_{\gamma\lambda\mu}(K))$ of eigenvalues of the operator
$H_{\gamma\lambda\mu}(K)$, for all $ K\in\T^2$, depending only on
the parameters  $\gamma,\lambda,\mu \in \R.$

\begin{theorem}\label{teo:bound_K}
Let $\gamma,\lambda,\mu\in \mathbb{R}$  and $ m, n=0,1,2,3,4,5,6,7$.
\begin{itemize}
\item[(i)]If $( \gamma,\lambda,\mu)\in \mathbb{G}_{mn} $ and $m+n<7$
then
$$
n_{-}\left(H_{ \gamma\lambda\mu}(K)\right)\geq m \quad \text{and}
\quad n_{+}\left( H_{ \gamma\lambda\mu}(K)\right)\geq n.
$$
\item[(ii)]If $( \gamma,\lambda,\mu)\in \mathbb{G}_{mn} $ and $m+n=7$,
then
$$
n_{-}\left(H_{ \gamma\lambda\mu}(K)\right)= m \quad \text{and} \quad
n_{+}\left( H_{ \gamma\lambda\mu}(K)\right)= n.
$$
\end{itemize}
\end{theorem}

\section{Auxiliary statements}\label{sec:auxiliary}

\subsection{The Lippmann--Schwinger operator}

Let $\{\alpha^{\mathrm{ees}}_{i},\,\,i=1,2,3,4\} $ be a  system of
vectors in $L^{2,\mathrm{ees}}(\T^2)$, with
\begin{equation}\label{ons2}
\begin{array}{cl}
&\alpha_{1}^{\mathrm{ees}}(p)=\dfrac{1}{2\pi},\quad  \quad  \quad
\alpha_{2}^{\mathrm{ees}}(p)=\dfrac{\cos p_{1}+\cos
p_{2}}{2\pi},\\
&\alpha_{3}^{\mathrm{ees}}(p)=\dfrac{\cos 2p_{1}+\cos 2p_{2}
}{2\pi}, \quad  \alpha_{4}^{\mathrm{ees}}(p)=\dfrac{\cos p_{1}\cos
p_{2} }{\pi}.
\end{array}
\end{equation}
One easily verifies by inspection that the vectors \eqref{ons2}  are
orthonormal in $L^{2,\mathrm{ees}}(\T^2)$. By using the orthonormal
system \eqref{ons2}  one obtains
\begin{align}\label{repr1}
&{V}^{\mathrm{ees}}_{\gamma\lambda\mu}{f}=\gamma({f},\alpha_{1}^{\mathrm{ees}})
\alpha_{1}^{\mathrm{ees}}+\frac{\lambda}{2}({f},\alpha_{2}^{\mathrm{ees}})
\alpha_{2}^{\mathrm{ees}}+
\frac{\mu}{2}({f},\alpha_{3}^{\mathrm{ees}})\alpha_{3}^{\mathrm{ees}}+
\frac{\mu}{2}({f},\alpha_{4}^{\mathrm{ees}})\alpha_{4}^{\mathrm{ees}}
\end{align}
where  $(\cdot,\cdot)$ is the inner product in
$L^{2,\mathrm{ees}}(\T^2).$ For any $z\in\C \setminus[0,\,8]$ we
define (the transpose of) the Lippmann-Schwinger operator (see.,
e.g., \cite{LSchwinger:1950}) as
\begin{equation*}\label{LSchw1}
{B}_{\gamma\lambda\mu}^{\mathrm{ees}}(0,z)=
-{V}_{\gamma\lambda\mu}^{\mathrm{ees}}{R}_0(0,z),
\end{equation*}
where  ${R}_0(0,z):= [{H}_0(0)-zI]^{-1},\,
z\in\mathbb{C}\setminus[0,\,8]$, is the resolvent of the operator
${H}_0(0)$.
\begin{lemma}\label{eigen-eigenvalue}
For every $\gamma,\lambda,\mu \in \R$  the number $z\in
\mathbb{C}\setminus [0,\,8]$ is an eigenvalue of
${H}_{\gamma\lambda\mu}^{\mathrm{ees}}(0)$ if and only if the number
$1$ is an eigenvalue for
 ${B}_{\gamma\lambda\mu}^{\mathrm{ees}}(0,z)$.
\end{lemma}
The proof of this lemma is quite standard (see., e.g.,
\cite{Albeverio:1988}). Thus, we skip~it.
\medskip

The representation \eqref{repr1} yields the equivalence of the
Lippmann-Schwinger equation
\begin{align*}
{B}_{\gamma\lambda\mu}^{\mathrm{ees}}(0,z){\varphi}={\varphi},
\,\,\,{\varphi} \in L^{2,\mathrm{ees}}(\T^2)
\end{align*}
to the following algebraic linear system in the variables
$x_i:=({\varphi},\alpha_{i}^{\mathrm{ees}}),\,\,i=1,2,3,4$:
\begin{equation}\label{system}
\left\lbrace\begin{array}{ccc}
[1+2\gamma a_{11}(z)]x_{1}+\lambda a_{12}(z)x_2+  \mu a_{13}(z)x_3+\mu a_{14}(z)x_4=0, \\
2\gamma  a_{12}(z) x_{1}+ [1+\lambda a_{22}(z)]x_{2}+ \mu a_{23}(z) x_{3}+\mu  a_{24}(z)x_{4}=0,\\
2\gamma  a_{13}(z) x_{1}+ \lambda a_{23}(z)x_{2}+ [1+\mu a_{33}(z) x_{3}]+\mu  a_{34}(z)x_{4}=0,\\
2\gamma  a_{14}(z) x_{1}+ \lambda a_{24}(z)x_{2}+ \mu a_{34}(z) x_{3}+[1+\mu  a_{44}(z)]x_{4}=0,\\
\end{array}\right.
\end{equation}
where
\begin{align}\label{def:functions}
 a_{ij}(z):=&  \frac{1}{2}\int_{\T^2} \frac{\alpha^{\mathrm{ees}}_i(p)
 \alpha^{\mathrm{ees}}_{j}(p)\d p}{\cE_0(p)-z},\,\, i,j=1,2,3,4.
\end{align}

\begin{proposition}\label{Rabcd}
For $z\in\mathbb{C}\setminus[0,8]$, the following relations hold:
\begin{align*}
& a_{ji}(z)=a_{ij}(z), \quad i,j=1,2,3,4,\\
&a_{12}(z)=\tfrac{(4-z)}{2}a_{11}(z)-\tfrac{1}{4},\\
&a_{22}(z)=\tfrac{(4-z)}{2}a_{12}(z)=\tfrac{(4-z)^2}{4}a_{11}(z)-
 \tfrac{(4-z)}{8}, \\
&a_{23}(z)=\tfrac{(4-z)}{2}a_{13}(z), \\
&a_{24}(z)=\tfrac{(4-z)}{2} a_{14}(z).
\end{align*}
\end{proposition}

\begin{proof}
The proof follows immediately from the definitions of the functions
$a_{ij}$ $ (i,j=1,2,3,4)$, see \eqref{def:functions}.
\end{proof}

Let
\begin{align}\label{determinant-ees}
&\Delta^{\mathrm{ees}}_{\gamma\lambda\mu}(z):=
\det[I-{B}_{\gamma\lambda\mu}^{\mathrm{ees}}(0,z)]\\
&=\begin{vmatrix}
1+2\gamma  a_{11}(z)&\lambda  a_{12}(z)&\mu  a_{13}(z)&\mu  a_{14}(z) \\
2\gamma  a_{12}(z)& 1+\lambda  a_{22}(z)&\mu  a_{23}(z)&\mu  a_{24}(z)\\
2\gamma  a_{13}(z)& \lambda  a_{23}(z)& 1+\mu  a_{33}(z)&\mu  a_{34}(z) \\
2\gamma  a_{14}(z)& \lambda  a_{24}(z)&\mu  a_{34}(z)&1+\mu
a_{44}(z)\nonumber
\end{vmatrix},
\end{align}
where $z\in\mathbb{C}\setminus[0,\,8]$.

The following lemma describes the relations between the function
$\Delta^\mathrm{ees}_{\gamma \lambda \mu  }(\cdot)$ and operator
$H_{\gamma \lambda \mu  }^\mathrm{ees}(0)$.%

\begin{lemma}\label{lem:det_zeros_vs_eigen}
Given $\gamma,\lambda,\mu \in \R$, a number $z\in
\mathbb{R}\setminus[0,8]$ is an eigenvalue of $H_{\gamma \lambda \mu
}^\mathrm{ees}(0)$  of multiplicity $m\ge1$ if and only if it is a
zero of $\Delta^\mathrm{ees}_{\gamma \lambda \mu }(\cdot)$ of
multiplicity $m.$ Moreover, in $\R\setminus [0,8]$ the function
$\Delta^\mathrm{ees}_{\gamma \lambda \mu  }(\cdot)$ has  at most
four zeros.
\end{lemma}

The proof of this statement is rather standard (cf e.g.
\cite{LBozorov:2009,LKhKhamidov:2021}) and we skip it.

\smallskip

Main properties of the functions \eqref{def:functions} and
asymptotic behavior  as $z\nearrow 0$ or $z\searrow 8$ is described
in the next asserion.

\begin{proposition}\label{prop:asymp_functions}
The functions  $a_{ij}(z)$ $(i,j=1,2,3,4)$ are real-valued for
$z\in\mathbb{R}\setminus[0,8]$, strictly increasing and positive in
$(-\infty,0)$, strictly increasing and negative in  $(8,+\infty)$.
Furthermore, the following asymptotics are valid:
\begin{align*}
a_{11}(z)=&\begin{cases}
-\frac{1}{8\pi}\ln(-z)+\frac{5\ln2}{8\pi}+o(1), \,\, \text{as} \quad z\nearrow0,  \\
\frac{1}{8\pi} \ln(z-8)-\frac{5\ln2}{8\pi}+o(1), \,\, \text{as}
\quad z\searrow 8,
\end{cases}
\end{align*}
\begin{align*}
a_{12}(z)=&\begin{cases}
-\frac{1}{4\pi}\ln(-z)+\frac{5\ln2-\pi}{4\pi}+o(1), \,\, \text{as} \quad z\nearrow0,  \\
\frac{1}{4\pi} \ln(z-8)-\frac{5\ln2-\pi}{4\pi}+o(1), \,\, \text{as}
\quad z\searrow 8,
\end{cases}
\end{align*}
\begin{align*}
a_{13}(z)=&\begin{cases}
-\frac{1}{4\pi}\ln(-z)+\frac{5\ln2-4\pi+8}{4\pi}+o(1), \,\, \text{as} \quad z\nearrow0,  \\
\frac{1}{4\pi} \ln(z-8)-\frac{5\ln2-4\pi+8}{4\pi}+o(1), \,\,
\text{as} \quad z\searrow 8,
\end{cases}
\end{align*}
\begin{align*}
a_{14}(z)=&\begin{cases}
-\frac{1}{4\pi}\ln(-z)+\frac{5\ln2-4}{4\pi}+o(1), \,\, \text{as} \quad z\nearrow0,  \\
\frac{1}{4\pi}\ln(z-8)-\frac{5\ln2-4}{4\pi}+o(1), \,\, \text{as}
\quad z\searrow 8,
\end{cases}
\end{align*}
\begin{align*}
a_{22}(z)=&\begin{cases}
-\frac{1}{2\pi}\ln(-z)+\frac{5\ln2-\pi}{2\pi}+o(1), \,\, \text{as} \quad z\nearrow0,  \\
\frac{1}{2\pi}\ln(z-8)-\frac{5\ln2-\pi}{2\pi}+o(1), \,\, \text{as}
\quad z\searrow 8,
\end{cases}
\end{align*}
\begin{align*}
a_{23}(z)=&\begin{cases}
-\frac{1}{2\pi}\ln(-z)+\frac{5\ln2-4\pi+8}{2\pi}+o(1), \,\, \text{as} \quad z\nearrow0,  \\
\frac{1}{2\pi}\ln(z-8)-\frac{5\ln2-4\pi+8}{2\pi}+o(1), \,\,
\text{as} \quad z\searrow 8,
\end{cases}
\end{align*}
\begin{align*}
a_{24}(z)=&\begin{cases}
-\frac{1}{2\pi}\ln(-z)+\frac{5\ln2-4}{2\pi}+o(1),
\,\, \text{as} \quad z\nearrow0,  \\
\frac{1}{2\pi}\ln(z-8)-\frac{5\ln2-4}{2\pi}+o(1), \,\, \text{as}
\quad z\searrow 8,
\end{cases}
\end{align*}
\begin{align*}
a_{33}(z)=&\begin{cases}
-\frac{1}{2\pi}\ln(-z)+\frac{5\ln2-20\pi+\frac{176}{3}}{2\pi}+o(1),
\,\, \text{as} \quad z\nearrow0,  \\
\frac{1}{2\pi}\ln(z-8)-\frac{5\ln2-20\pi+\frac{176}{3}}{2\pi}+o(1),
\,\, \text{as} \quad z\searrow 8,
\end{cases}
\end{align*}
\begin{align*}
a_{34}(z)=&\begin{cases}
-\frac{1}{2\pi}\ln(-z)+\frac{5\ln2+4\pi-\frac{52}{3}}{2\pi}+o(1),
\,\, \text{as} \quad z\nearrow0,  \\
\frac{1}{2\pi} \ln(z-8)-\frac{5\ln2+4\pi-\frac{52}{3}}{2\pi}+o(1),
\,\, \text{as} \quad z\searrow 8,
\end{cases}
\end{align*}
\begin{align*}
a_{44}(z)=&\begin{cases}
-\frac{1}{2\pi}\ln(-z)+\frac{5\ln2-2\pi+\frac{8}{3}}{2\pi}+o(1),
\,\, \text{as} \quad z\nearrow0,  \\
\frac{1}{2\pi}\ln(z-8) -\frac{5\ln2-2\pi+\frac{8}{3}}{2\pi}+o(1),
\,\, \text{as} \quad z\searrow 8,
\end{cases}
\end{align*}
where \ $\ln(-z)$ and $\ln(z-8)$ denote those  branches of the
corresponding analytic functions that are real for $-z>0$ and $z>8$,
respectively.

\end{proposition}

\begin{proof}
Let us prove the positivity of the functions $a_{ij}$,
$i,j=1,2,3,4,$  in $(-\infty,0)$.

The positivity of $a_{11}(z),a_{22}(z),a_{33}(z),a_{44}(z)$ for
$z<0$ is obvious. In order to prove the positivity of $a_{13}(z)$,
$z<0$, we first consider an auxiliary integral
$$
J(u)= \int_{-\pi}^{\pi}\frac{\cos 2p \,\d p}{u-\cos p},\quad u>1.
$$
Clearly,
\begin{equation}
\label{Int1} J(u)=\int_{-\pi}^{\pi}\frac{\cos 2p \,\d p}{u-\cos
p}=\int_{-\pi}^{0}\frac{\cos 2p \,\d p}{u-\cos
p}+\int_{0}^{\pi}\frac{\cos 2p \,\d p}{u-\cos p}.
\end{equation}
 Changing the variable $p=q-\pi$ in the first integral on the
r.h.s. part of \eqref{Int1}, one obtains
\begin{align*}
J(u)=&\int_{0}^{\pi}\frac{\cos 2(q-\pi) \,\d q}{u-\cos
(q-\pi)}+\int_{0}^{\pi}\frac{\cos 2p \,\d p}{u-\cos p}
=\int_{0}^{\pi}\frac{\cos 2q \,\d q}{u+\cos
q}+\int_{0}^{\pi}\frac{\cos 2p \,\d
p}{u-\cos p}\\
=&\int_{0}^{\pi}\frac{2u\cos 2p \,\d p}{u^2-\cos^2 p}
=\int_{0}^{\frac{\pi}{2}}\frac{2u\cos 2p \,\d p}{u^2-\cos^2
p}+\int_{\frac{\pi}{2}}^{\pi}\frac{2u\cos 2p \,\d p}{u^2-\cos^2 p}.
\end{align*}
Finally, by making the substitution $p = q + \frac{\pi}{2}$ in the
very last term, we find
\begin{align*}
J(u)=&\int_{0}^{\frac{\pi}{2}}\frac{2u\cos 2p \,\d p}{u^2-\cos^2
p}-\int_{0}^{\frac{\pi}{2}}\frac{2u\cos 2q \,\d
q}{u^2-\sin^2 q} \\
=&2u\int_{0}^{\frac{\pi}{2}}\frac{\cos^2 2p \,\d p}{(u^2-\cos^2
p)(u^2-\sin^2 p)}\quad \text{for any }u>1.
\end{align*}
and, thus,
\begin{equation}
\label{Jub0} J(u)= \int_{-\pi}^{\pi}\frac{\cos 2p \,\d p}{u-\cos
p}>0\quad \text{for any }u>1.
\end{equation}

Now observe that $u=\frac{4-z}{2}-\cos q>1$ for any $q\in\mathbb{R}$
and any $z<0$. Combining this observation with \eqref{Jub0} one
concludes that
\begin{align*}
a_{13}(z)=&\frac{1}{8\pi^2}\int_{\T}\int_{\T}\frac{(\cos 2p_1+\cos
2p_2) }{4-z-2\cos p_1-2\cos p_2}\,\d
p_1\d p_2\\
=&\frac{1}{8\pi^2}\int_{\pi}^{\pi}\left( \int_{-\pi}^{\pi}\frac{\cos
2p_1 \,\d p_1}{\bigl(\frac{4-z}{2}-\cos p_2\bigr)-\cos p_1}\right)\d
p_2>0.
\end{align*}
 Similarly,
\begin{align*}
a_{14}(z)=&\frac{1}{4\pi^2}\int_{\pi}^{\pi}\left(
\int_{-\pi}^{\pi}\frac{\cos p_1\cos p_2 \,\d
p_1}{4-z-2\cos p_1-2\cos p_2}\right)\d p_2\\
=&\frac{1}{4\pi^2}\int_{-\pi}^{0}\left( \int_{-\pi}^{0}\frac{\cos
p_1\cos p_2 \,\d
p_1}{4-z-2\cos p_1-2\cos p_2}\right)\d p_2\\
&+\frac{1}{4\pi^2}\int_{-\pi}^{0}\left( \int_{0}^{\pi}\frac{\cos
p_1\cos p_2 \,\d
p_1}{4-z-2\cos p_1-2\cos p_2}\right)\d p_2\\
&+\frac{1}{4\pi^2}\int_{0}^{\pi}\left( \int_{-\pi}^{0}\frac{\cos
p_1\cos p_2 \,\d
p_1}{4-z-2\cos p_1-2\cos p_2}\right)\d p_2\\
&+\frac{1}{4\pi^2}\int_{0}^{\pi}\left( \int_{0}^{\pi}\frac{\cos
p_1\cos p_2 \,\d p_1}{4-z-2\cos p_1-2\cos p_2}\right)\d p_2.
\end{align*}
In the last four integrals, by making the respective change of
variables
\begin{align*}
&p_1=q_1-\pi, \quad p_2=q_2-\pi, \\
&p_1=q_1, \quad \quad \quad p_2=q_2-\pi, \\
&p_1=q_1-\pi, \quad   p_2=q_2, \\
&p_1=q_1, \quad \quad \quad  p_2=q_2,
\end{align*}
we obtain
\begin{align*}
a_{14}(z) =&\frac{1}{8\pi^2}\int_{0}^{\pi}\left(
\int_{0}^{\pi}\frac{\cos q_1\cos q_2 \,\d q_1}{\frac{4-z}{2}+\cos
q_1+\cos q_2}\right)\d q_2\\
&- \frac{1}{8\pi^2}\int_{0}^{\pi}\left( \int_{0}^{\pi}\frac{\cos
q_1\cos q_2 \,\d
q_1}{\frac{4-z}{2}-\cos q_1+\cos q_2}\right)\d q_2\\
&-\frac{1}{8\pi^2}\int_{0}^{\pi}\left( \int_{0}^{\pi}\frac{\cos
q_1\cos q_2 \,\d q_1}{\frac{4-z}{2}+\cos q_1-\cos q_2}\right)\d
q_2\\
&+\frac{1}{8\pi^2}\int_{0}^{\pi}\left( \int_{0}^{\pi}\frac{\cos
q_1\cos q_2 \,\d q_1}{\frac{4-z}{2}-\cos q_1-\cos q_2}\right)\d q_2.
\end{align*}
By grouping separately the first and third terms and then the second
and forth terms one arrives at
\begin{align*}
a_{14}(z) =&-\frac{1}{4\pi^2}\int_{0}^{\pi}\left(
\int_{0}^{\pi}\frac{\cos q_1\cos^2 q_2 \,\d
q_1}{(\frac{4-z}{2}+\cos q_1)^2-\cos^2 q_2}\right)\d q_2\\
&+ \frac{1}{4\pi^2}\int_{0}^{\pi}\left( \int_{0}^{\pi}\frac{\cos
q_1\cos^2 q_2 \,\d
q_1}{(\frac{4-z}{2}-\cos q_1)^2-\cos^2 q_2}\right)\d q_2\\
=&\frac{1}{2\pi^2}\int_{0}^{\pi}\left(
\int_{0}^{\pi}\frac{(4-z)\cos^2 q_1\cos^2 q_2 \,}{\left[
(\frac{4-z}{2}-\cos q_1)^2-\cos^2 q_2\right]
\left[(\frac{4-z}{2}+\cos q_1)^2-\cos^2 q_2\right]}\d q_1\right)\d
q_2.
\end{align*}
Obviously, the expression under the integration sign in the last
integral is always positive which means that $a_{14}(z)>0$ for any
$z<0$.\

The positivity of the functions $a_{12},  a_{23}, a_{24}$ follows
from Proposition \ref{Rabcd} and positivity of the functions
$a_{22},  a_{13}, a_{14},$ respectively.

The strict monotonicity of the functions $a_{ij}$, $i,j=0,1,2,3,4,$
in the intervals $(-\infty,0)$ and $(8,+\infty)$ is proven
technically in the same way as their respective sign definiteness.

The asymptotic relations for the functions $a_{11}$, $a_{12}$,
$a_{22}$ and $a_{14}(\cdot)$ have been established in
\cite{LKhKhamidov:2021}.   The remaining asymptotic relations are
proven by using Proposition \ref{Rabcd} and the following limit
equalities:
\begin{align*}
\lim\limits_{z\nearrow 0} \Big(
a_{13}(z)-2a_{11}(z)\Big)=&-\frac{1}{2\pi^2}\int_{\T^2}\frac{\sin^2
p_1 \,\d
p}{\cE_0(p)} =-\frac{\pi-2}{\pi}, \\
\lim\limits_{z\nearrow 0}\Big( a_{33}(z)-4a_{11}(z)\Big)=&
-\frac{1}{\pi^2}\int_{\T^2}\frac{\sin^2 p_1\cos^2 p_2 \,\d
p}{\cE_0(p)}- \frac{1}{\pi^2}\int_{\T^2}\frac{\sin^2 p_1\cos^2 p_1
\,\d
p}{\cE_0(p)}\\
&=\frac{88-30\pi}{3\pi}, \\
\lim\limits_{z\nearrow 0} \Big( a_{44}(z)-4a_{11}(z)\Big)=&
-\frac{1}{2\pi^2}\int_{\T^2}\frac{\sin^2 p_1\cos^2 p_2 \,\d
p}{\cE_0(p)}- \frac{1}{2\pi^2}\int_{\T^2}\frac{\sin^2 p_1 \,\d
p}{\cE_0(p)}\\
&=-\frac{3\pi-4}{3\pi}& \\
\lim\limits_{z\nearrow 0} \Big( a_{34}(z)-2a_{14}(z)\Big) =&-
\frac{1}{\pi^2}\int_{\T^2}\frac{\sin^2 p_1\cos p_1 \cos p_2 \,\d
p}{\cE_0(p)}=2-\frac{20}{3\pi}.
\end{align*}
The asymptotic expressions at $z=8$ are derived in a similar way,
and therefore, we omit the corresponding calculations.
\end{proof}

\begin{lemma}\label{lem:asymp_det}
For any $\gamma,\lambda,\mu \in \R$ the function
$\Delta^\mathrm{ees}_{\gamma \lambda \mu }(z)$ is holomorphic in
$z\in\C\setminus[0,8]$. Furthermore, this function is real analytic
for $z\in\R\setminus[0,8]$ and possesses the following asymptotics:
\begin{align*}
&\lim\limits_{z\rightarrow {\pm\infty}}
\Delta^\mathrm{ees}_{\gamma\lambda \mu }(z)=1,\\
&\Delta^\mathrm{ees}_{\gamma \lambda \mu }(z) =
-\tfrac{1}{4\pi}Q^{-}(\gamma ,\lambda ,\mu  )\ln(-z)+
Q^{-}_{0}(\gamma ,\lambda ,\mu )+o(1),\,\, z \nearrow 0 \quad(z<0),    \\
&\Delta^\mathrm{ees}_{\gamma \lambda \mu }(z)
=-\tfrac{1}{4\pi}Q^{+}(\gamma ,\lambda ,\mu
)\ln(z-8)+Q^{+}_{0}(\gamma ,\lambda ,\mu )+o(1), \, z \searrow 8
\quad (z>8),
\end{align*}
where  $Q^{\pm}(\gamma,\lambda,\mu)$ is given by
\eqref{polynomial:Q}.
 \end{lemma}

\begin{proof}
By taking into account \eqref{EKmin} and \eqref{EKmax}, the
holomorhy of $\Delta^\mathrm{ees}_{\gamma \lambda \mu }(z)$ for
$z\in\C\setminus[0,8]$ follows immediately from the same property of
the functions \eqref{def:functions}. Similarly, the real-valuedness
of $\Delta^\mathrm{ees}_{\gamma \lambda \mu }(z)$ for
$z\in\R\setminus[0,8]$ is implied by the real-valuedness of
\eqref{def:functions} for those $z$. In its turn, the first
asymptotic relation follows from the Lebesgue dominated  convergence
theorem. The two remaining relations are proven by Proposition
\ref{prop:asymp_functions}.
\end{proof}

The following two lemmas establish the number and location of
eigenvalues of the operator $H^{\mathrm{ees}}_{\gamma\lambda0}(0)$.

\begin{lemma}\label{lem:simple1}
Let $(\gamma,\lambda)\in\R^2.$

\begin{itemize}
\item[(i)] If $\gamma\lambda-4\lambda -2\gamma\ge0$ and $\lambda<2$,
then the operator $H^{\mathrm{ees}}_{\gamma\lambda0}(0)$ has no
eigenvalues in $(8,+\infty).$

\item[(ii)] If $\gamma\lambda-4\lambda -2\gamma<0$ (and $\lambda\in\R$)
or $\gamma\lambda-4\lambda -2\gamma=0$ (and $\lambda>2$), then the
operator $H^{\mathrm{ees}}_{\gamma\lambda0}(0)$ has a single
eigenvalue in $(8,+\infty)$ and this eigenvalue is simple.

\item[(iii)] If $\gamma\lambda-4\lambda -2\gamma>0$ and $\lambda>2$,
then the operator $H^{\mathrm{ees}}_{\gamma\lambda0}(0)$ has exactly
two simple eigenvalues in $(8,+\infty).$
\end{itemize}

\end{lemma}

\begin{lemma}\label{lem:simple2}
Let $(\gamma,\lambda)\in\R^2.$

\begin{itemize}
\item[(i)] If $\gamma\lambda+4\lambda +2\gamma\ge0$ and $\lambda>-2$,
then the operator $H^{\mathrm{ees}}_{\gamma\lambda0}(0)$ has no
eigenvalues in $(-\infty,0).$

\item[(ii)] If $\gamma\lambda+4\lambda+2\gamma<0$
(and $\lambda\in\R$) or $\gamma\lambda+4\lambda+2\gamma=0$ (and
$\lambda<-2$),  then the operator
$H^{\mathrm{ees}}_{\gamma\lambda0}(0)$ has a single eigenvalue in
$(-\infty,0)$ and this etigenvalue is simple.

\item[(iii)] If $\gamma\lambda+4\lambda + 2\gamma>0$
and $\lambda<-2$, then the operator
$H^{\mathrm{ees}}_{\gamma\lambda0}(0)$ has two simple eigenvalues in
$(-\infty,0).$
\end{itemize}

\end{lemma}

Both Lemmas \ref{lem:simple1} and  \ref{lem:simple2} are proven  in
\cite{LKhKhamidov:2021} (see \cite[Lemmas 4.6 and 4.7,
respectively]{LKhKhamidov:2021}).

\section{Proof of the main results}\label{sec:proofs}

This section is dedicated to the proof of Theorems
\ref{teo:constant} --  \ref{teo:bound_K}.
\medskip

\noindent \textit{Proof of Theorem \ref{teo:constant}}. We perform
the proof for the case where
$\cC\in\{\cC^{-}_{0},\cC^{-}_{1},\cC^{-}_{2}, \cC^{-}_{3},
\cC^{-}_{4}\}$. The cases with $\cC=\cC_\alpha^+$,
$\alpha=0,1,2,3,4,$ with the opposite index ``$+$'' may be treated
similarly and we skip the corresponding proofs.

Notice that by Remark \ref{Rem_Q_sigdef} the function $Q^{-}(\gamma,
\lambda, \mu)$ is sign-definite on $\cC$. Without loss of generality
one may assume that, say,
$$
Q^{-}(\gamma, \lambda, \mu)>0\quad \text{for any
\,}(\gamma,\lambda,\mu)\in\cC.
$$
Bearing in mind that $Q^{-}(\gamma,\lambda,\mu) =
-4\pi\,\lim\limits_{z \nearrow 0}
\displaystyle\frac{\Delta^\text{ees}_{\gamma\lambda\mu}(z)}{\ln(-z)}$
\, (see Lemma \ref{lem:asymp_det}) one concludes that there is a
number $\delta>0$, perhaps, depending on
$(\gamma,\lambda,\mu)\in\cC$, such that
\begin{equation}
\label{Del_poz}
\Delta^\text{ees}_{\gamma\lambda\mu}(z)>0\quad\text{whenever\,\, }
z\in(-\delta,0).
\end{equation}
Since the function $\Delta^\text{ees}_{\gamma\lambda\mu}(z)$,
$(\gamma,\lambda,\mu)\in\cC$, is real analytic in $z<0$ and
\begin{equation}
\label{Del_one}
\lim_{z\to-\infty}\Delta^{\text{ees}}_{\gamma\lambda\mu}(z) = 1
\end{equation}
(again see  Lemma \ref{lem:asymp_det}), one then derives that
$\Delta^{\text{ees}}_{\gamma\lambda\mu}(z)$ may only have a finite
number of zeros in $z\in(-\infty,0)$ and all these zeros are
confined in the interval $(b_1,b_2)$ with some negative $b_1$ and
$b_2$, $-\infty<b_1<b_2<0$, possibly depending on $\gamma$,
$\lambda$, and $\mu$.

Now recall that $\cC$ is a connected open set in $\R^3$ and pick up
two arbitrary different points $(\gamma_0,\lambda_0,\mu_0)$ and
$(\gamma_1,\lambda_1,\mu_1)$ in $\cC$. Let $\Gamma$ be a simple
closed curve that connects the points $(\gamma_0, \lambda_0, \mu_0)$
and $(\gamma_1, \lambda_1, \mu_1)$, and lies entirely within the
region $\mathcal{C}$. Since the curve $\Gamma\subset\cC$ is a
compact and the function $\Delta^{\text{ees}}_{\gamma\lambda\mu}(z)$
is real analytic in $(\gamma,\lambda,\mu)\in\cC$ and $z<0$, the
minimum
$$
m(z):=\min_{(\gamma,\lambda,\mu)\in\Gamma}\Delta^{\text{ees}}_{\gamma\lambda\mu}(z),
\quad z<0,
$$
is well defied and continuously depends on $z\in(-\infty,0)$.
Furthermore, from \eqref{Del_poz} and \eqref{Del_one} it follows
that there are numbers $B_1$ and $B_2$, $-\infty<B_1<B_2<0$, perhaps
depending on $\Gamma$, such that
$$
m(z)>0\quad \text{ for any\,\ }z\in(-\infty,B_1]\cup[B_2,0).
$$
Obviously, this means that for any $(\gamma,\lambda,\mu)\in\Gamma$
all the roots of the function
$\Delta^{\text{ees}}_{\gamma\lambda\mu}(z)$, $z<0$, are located in
the interval $(B_1,B_2)$. It then follows that there is a real
number $\varepsilon>0$ such that
$|\Delta^{\text{ees}}_{\gamma\lambda\mu}(z)|>\varepsilon$ whenever
$z$ lies on the circumference
$W_\eta(z_0):=\bigl\{z\in\C\,\bigl|\,\,|z-z_0|=\eta\bigr\}$ where
$z_0:=\frac{B_1+ B_2}{2}$ and $\eta:=\frac{B_2 - B_1}{2}$.

As a next step, we prove that there is a number $\beta>0$ such that
for any $(\gamma, \lambda, \mu)\in\Gamma$ fulfilling inequality
\begin{equation}
\label{DelBet} \|(\gamma, \lambda, \mu)-(\gamma_0, \lambda_0,
\mu_0)\|<\beta
\end{equation}
the bound
\begin{equation}
|\Delta^{\text{ees}}_{\gamma \lambda\mu} (z) - \Delta_{\gamma_0
\lambda_0 \mu_0}^{\text{ees}} (z)| < \epsilon \label{DelEps}
\end{equation}
holds simultaneously for all $z\in W_\eta(z_0)$. We prove this by
contradiction. Assume that the opposite holds, that is, there is
$z\in W_\eta(z_0)$ such that for any $\beta>0$ one may pick up a
point $(\gamma, \lambda, \mu)\in \Gamma$ satisfying both the
inequalities \eqref{DelBet} and
\begin{equation}
|\Delta^{\text{ees}}_{\gamma \lambda\mu} (z) - \Delta_{\gamma_0
\lambda_0 \mu_0}^{\text{ees}} (z)| \geq \epsilon. \label{NDelEps}
\end{equation}
But the inequalities \eqref{DelBet} and \eqref{NDelEps} together
with the arbitrariness of the (positive) bound $\beta$ mean that
$(\gamma_0, \lambda_0, \mu_0)$ is a discontinuity point of the
determinant $\Delta^{\text{ees}}_{\gamma \lambda\mu}(z)$ considered
as a function of $\gamma, \lambda,\mu$ at the fixed $z$. This
contradicts the analyticity the function
$\Delta^{\text{ees}}_{\gamma \lambda\mu}(z)$ in $\gamma, \lambda,
\mu\in\R$ and, hence, its continuity  in $(\gamma, \lambda,
\mu)\in\Gamma$.

Therefore, the existence of a number $\beta>0$ such that for any
$(\gamma,\lambda,\mu)\in\Gamma$ satisfying \eqref{DelBet} the bound
\eqref{DelEps} is valid for any $z\in W_\eta(z_0)$, has been
established.

Then by Rouche's theorem the number of roots of the function
$\Delta^{\mathrm{ees}}_{\gamma\lambda\mu}(z)$ in $z\in(B_1, B_2)$
(counting multiplicities) is the same for any point
$(\gamma,\lambda,\mu)\in\Gamma$ satisfying \eqref{DelBet} and,
hence, the same for any $(\gamma,\lambda,\mu)\in\Gamma$. Taking into
account the arbitrariness of the simple close curve
$\Gamma\subset\cC$, one then easily derives that the number of roots
of the function $\Delta^{\mathrm{ees}}_{\gamma\lambda\mu}(z)$ in
$z<0$  (counting multiplicities) is the same for any
$(\gamma,\lambda,\mu)\in\cC$, which completes the proof. \qed

\medskip

\noindent\textit{Proof of Theorem \ref{teo:ee,s}}. We give the
proofs successfully for the cases $\alpha=0,1,2,3,4$ (and sign
``$-$'').

Pick up the point $(1,0,0)\in \cC^{-}_{0}$. By Lemma
\ref{lem:simple2}\,(i), the operator  $H^\mathrm{ees}_{100}(0)$ has
no eigenvalues. Then Theorem \ref{teo:constant} yields that the
operator  $H^\mathrm{ees}_{\gamma\lambda\mu}(0)$ has no  eigenvalues
below the essential spectrum for all $(\gamma,\lambda,\mu)\in
\cC^{-}_{0}$.

By taking into account definition \eqref{polynomial:Q} and the
second equality in \eqref{def:ees_sets} one concludes  that
$(0,-1,0)\in\cC^{-}_{1}$. Meanwhile, from Lemma
\ref{lem:simple2}\,(ii) it follows that $H^\mathrm{ees}_{0(-1)0}(0)$
has exactly one eigenvalue below the essential spectrum. Then by
Theorem \ref{teo:constant} the same property holds for
$H^\mathrm{ees}_{\gamma\lambda\mu}(0)$ with arbitrary
$(\gamma,\lambda,\mu)\in \cC^{-}_{1}$.

By the third equality in \eqref{def:ees_sets} we have
$(-5,-11,0)\in\cC^{-}_{2}$. One also notices that by Lemma
\ref{lem:simple2}\,(iii) the operator
$H^\mathrm{ees}_{(-5)(-11)0}(0)$ has exactly two simple eigenvalues.
Then Theorem \ref{teo:constant} yields the same property for
$H^\mathrm{ees}_{\gamma\lambda\mu}(0)$ whenever
$(\gamma,\lambda,\mu)\in \cC^{-}_{2}$.

From now on we assume that $(\gamma,\lambda,\mu)\in \cC^{-}_{4}$.
Recall that the last rows in \eqref{def:regions_D} and
\eqref{def:ees_sets} read together as
\begin{align}\label{ineq_glm1}
&Q^{-}(\gamma,\lambda,\mu)=\left[\lambda Q^{-}_{0}(\mu)+
Q^{-}_1(\mu) \right](\gamma-\gamma^{-}(\lambda,\mu))>0, \,\, \\
&\lambda Q^{-}_{0}(\mu)+ Q^{-}_1(\mu)<0, \quad  \mu<\mu_0^-<0.
\label{mumu0}
\end{align}
The equalities   \eqref{polynomial:q0}, \eqref{polynomial:q1} and
inequality $\mu<\mu_0^- $ in \eqref{mumu0} imply that
 \begin{equation}\label{Q1Q0geo}
 Q^{-}_{0}(\mu)>0 \quad \text{and} \quad Q^{-}_{1}(\mu)>0.
\end{equation}
From  $\lambda Q^{-}_{0}(\mu)+ Q^{-}_1(\mu)<0$  and  \eqref{Q1Q0geo}
it follows that
\begin{equation}
\label{l_negative} \lambda<0 \quad\text{whenever
}(\gamma,\lambda,\mu)\in \cC^{-}_{4}.
\end{equation}

By the definition \eqref{polynomial:Q} we have
\begin{equation*}
Q^{-}(\gamma, \lambda, \mu)=(\gamma + 4)\Big(\lambda Q^{-}_{0}(\mu)
+ Q^{-}_1(\mu)\Big)-8Q^{-}_{0}(\mu).
\end{equation*}
Then by combining the inequalities  \eqref{ineq_glm1} and
\eqref{mumu0} with the first inequality in \eqref{Q1Q0geo} one
derives that
\begin{equation}\label{g_negative}
\gamma+4<0, \,\, \text{i.e.} \,\, \gamma<0 \quad\text{whenever
}\,(\gamma,\lambda,\mu)\in \cC^{-}_{4}.
\end{equation}

Observe that the function
$$
\delta(z):=1+\mu  a_{44}(z), \quad \mu<0,
$$
with $a_{44}(z)$ given in \eqref{def:functions} is continuous and
strictly decreasing in $z\in(-\infty,0)$. From the explicit
definition of $a_{44}$ in \eqref{def:functions} it follows that
$$
\lim\limits_{z\rightarrow -\infty}\delta(z)=1.
$$
At the same time, the asymptotic expression for $a_{44}(z)$ in
Proposition \ref{prop:asymp_functions} implies that
$$
\lim\limits_{z\nearrow 0}\delta(z)=-\infty .
$$
Therefore the function $\delta(z)= 1+\mu  a_{44}(z)$  has exactly
one zero $z_{11}$ within the half-axis $(-\infty,0)$ and, thus,
\begin{equation}\label{ineq1}
\begin{array}{cl}
& 1+\mu  a_{44}(z)>0 \quad \text{if} \quad z< z_{11},\\
& 1+\mu  a_{44}(z)<0  \quad \text{if} \quad z_{11}< z<0.
\end{array}
\end{equation}
Notice that the equality $1+\mu  a_{44}(z_{11})=0$ implies that
\begin{align}
\nonumber \Delta^\mathrm{ees}_{00\mu}(z_{11})&=\left(1+\mu
a_{33}(z_{11})\right)\left(1+\mu  a_{44}(z_{11})\right)-\mu^2  (a_{34}(z_{11}))^2\\
\label{ineq2} &=-\mu^2 (a_{34}(z_{11}))^2< 0.
\end{align}
Meanwhile, for $\mu<\mu_0^-\,\,(<0)$ we have
\begin{align}\label{Qmm}
Q^{-}(0,0,\mu)&=\tfrac{6(16-5\pi)}{\pi}\mu\left[\mu+
\tfrac{2\pi}{3(16-5\pi)}\right]>0.
\end{align}
The inequalities \eqref{ineq2},  \eqref{Qmm} and  Lemma
\ref{lem:asymp_det}  yield that
$$
\lim\limits_{z\rightarrow -\infty}\Delta^\mathrm{ees}_{00\mu}(z)=1,
\quad  \Delta^\mathrm{ees}_{00\mu}(z_{11})< 0  \quad \text{and}
\quad \lim\limits_{z\nearrow
0}\Delta^\mathrm{ees}_{00\mu}(z)=+\infty .
$$
This means that there exist real numbers $z_{21}$ and $z_{22}$ such
that
\begin{equation}\label{loc_roots}
z_{21}< z_{11}< z_{22}<0
\end{equation}
and
\begin{align}\label{root2}
\Delta^\mathrm{ees}_{00\mu}(z_{21})=\Delta^\mathrm{ees}_{00\mu}(z_{22})=0.
\end{align}

The equality \eqref{repr1} implies that  the operator
$V^\mathrm{ees}_{00\mu}$ has rank at most two. Therefore, by the
minimax principle (see \cite{RSimon:IV}, page 85) the operator
$H^\mathrm{ees}_{00\mu}(0)$ has at most two eigenvalues below zero.
By the first statement in Lemma \ref{lem:det_zeros_vs_eigen}, the
isolated eigenvalues of the operator $H^\mathrm{ees}_{00\mu}(0)$
coincide with zeros of the function
$\Delta^\mathrm{ees}_{00\mu}(z)$. This yields that the function
$\Delta^\mathrm{ees}_{00\mu}(z)$ has at most two zeros in
$\R\setminus [0,8]$. By \eqref{root2} the function
$\Delta^\mathrm{ees}_{00\mu}(z)$ has exactly two zeros ($z_{21}$ and
$z_{22}$), lying below the essential spectrum. Therefore
\begin{equation}\label{2-ineq}
\begin{array}{cl}
&\Delta^\mathrm{ees}_{00\mu}(z)> 0 \quad \text{if} \quad z<  z_{21};\\
&\Delta^\mathrm{ees}_{00\mu}(z)< 0 \quad \text{if} \quad z_{21}< z< z_{22};\\
&\Delta^\mathrm{ees}_{00\mu}(z)> 0 \quad \text{if} \quad z_{22}<
z<0.
\end{array}
\end{equation}

The equalities \eqref{root2} yield the following relations

\begin{equation}\label{ineq4}
\begin{array}{cl}
 (1+\mu  a_{33}(z_{21}))(1+\mu  a_{44}(z_{21}))=&\mu^2  (a_{34}(z_{21}))^2> 0 , \\
(1+\mu  a_{33}(z_{22}))(1+\mu  a_{44}(z_{22}))=&\mu^2
(a_{34}(z_{22}))^2> 0.
\end{array}
\end{equation}
By  \eqref{ineq1} and  \eqref{loc_roots} we have $1+\mu
a_{44}(z_{21})>0$ and $1+\mu a_{44}(z_{22})<0$. Therefore
\eqref{ineq4} yields
\begin{equation}\label{ineq5}
\begin{array}{cl}
&1+\mu  a_{33}(z_{21})> 0 \quad \text{and} \quad
1+\mu  a_{44}(z_{21})> 0, \\
&1+\mu  a_{33}(z_{22})< 0  \quad \text{and} \quad 1+\mu
a_{44}(z_{22})< 0.
\end{array}
\end{equation}
Moreover, in view of the positivity of $a_{34}(z_{21})$ and
$a_{34}(z_{22})$ (see Proposition \ref{prop:asymp_functions}), by
\eqref{ineq4} and \eqref{ineq5} one concludes that

\begin{equation}\label{eq1}
\sqrt{ 1+\mu  a_{33}(z_{21})}\sqrt{ 1+\mu  a_{44}(z_{21})}=-\mu
a_{34}(z_{21})
\end{equation}
and
\begin{equation}\label{eq2}
\sqrt{ -[1+\mu  a_{33}(z_{22})]}\sqrt{ -[1+\mu a_{44}(z_{22})]}=-\mu
a_{34}(z_{22}).
\end{equation}
Now, by using the explicit representation  \eqref{determinant-ees}
for $\Delta^\mathrm{ees}_{0 \lambda \mu} (z)$, the positivity of
$a_{ij}(z)$ $(i,j=1,2,3,4)$ for $z\in(-\infty,0)$ (see Proposition
\ref{prop:asymp_functions}), the inequalities $\mu,\,\lambda<0$  and
equality \eqref{eq1} one arrives at the following conclusions:
\begin{align}
\nonumber \Delta^\mathrm{ees}_{0 \lambda \mu }(z_{21})=& -\lambda\mu
\big[\sqrt{ 1+\mu  a_{44}(z_{21})}a_{23}(z_{21})\\\nonumber &+\sqrt{
1+\mu
a_{33}(z_{21})}a_{24}(z_{21})\big]^2\\
\label{ineq6}
< &0,\\
\nonumber \Delta^\mathrm{ees}_{0 \lambda \mu }(z_{22})=& \lambda\mu
\big[\sqrt{ -[1+\mu  a_{44}(z_{22})]}a_{23}(z_{22})\\\nonumber
&+\sqrt{ -[1+2\mu a_{33}(z_{22})]}a_{24}(z_{22})\big]^2
\\
> & 0.
\label{D0lm_in32>0}
\end{align}

For all $(\gamma,\lambda,\mu)\in  \cC^{-}_{4}$ the inequalities
$\lambda Q^{-}_{0}(\mu)+Q^{-}_{1}(\mu)<0$ in \eqref{ineq_glm1} and
$Q^{-}_0(\mu)>0$ in \eqref{Q1Q0geo} yield that
\begin{align*}
Q^{-}(0 ,\lambda ,\mu)=&4( \lambda Q^{-}_{0}(\mu)+Q^{-}_{1}(\mu)) -
8  Q^{-}_0(\mu) <0.
\end{align*}
By  Lemma \ref{lem:asymp_det} this inequality gives
\begin{equation}\label{bound_Delta}
\lim\limits_{z\rightarrow -\infty}\Delta^\mathrm{ees}_{0 \lambda \mu
}(z)=1
 \quad \text{and} \quad  \lim\limits_{z\nearrow 0}
 \Delta^\mathrm{ees}_{0 \lambda \mu }(z)=-\infty.
\end{equation}
The relations \eqref{ineq6},   \eqref{D0lm_in32>0} and
\eqref{bound_Delta}  imply there existence of real numbers $z_{31}$,
$z_{32}$  and $z_{33}$ such that
\begin{equation}\label{loc_roots2}
 z_{31}< z_{21}< z_{32}< z_{22}< z_{33}<0
\end{equation}
and
\begin{equation}\label{delta_roots3}
\Delta^\mathrm{ees}_{0 \lambda \mu }(z_{31})= \Delta^\mathrm{ees}_{0
\lambda \mu }(z_{32})= \Delta^\mathrm{ees}_{0 \lambda \mu
}(z_{33})=0.
\end{equation}

The equality \eqref{repr1} implies that  the operator
$V^\mathrm{ees}_{0\lambda\mu}$ has rank at most three. Therefore,
again by the minimax principle (see \cite{RSimon:IV}, page 85), the
operator $H^\mathrm{ees}_{0\lambda\mu}(0)$ has at most three
eigenvalues outside the essential spectrum. The relation between the
eigenvalues of the operator $H^\mathrm{ees}_{0\lambda\mu}(0)$ and
zeros of the function $\Delta^\mathrm{ees}_{0\lambda\mu}(z)$ (see
Lemma \ref{lem:det_zeros_vs_eigen}) yields that the function
$\Delta^\mathrm{ees}_{0\lambda\mu}(z)$ has at most three zeros in
$\R\setminus [0,8]$. Hence, the function
$\Delta^\mathrm{ees}_{0\lambda\mu}(z)$ has exactly three zeros (
$z_{31}$, $z_{32}$ and $z_{33}$), all of them lying below the
essential spectrum.

Let
\begin{equation}\label{def:Aij}
\begin{array}{cl}
&A_{12}(z):=\det\begin{pmatrix}
\lambda  a_{23}(z)& \mu  a_{34}(z)\\
\lambda  a_{24}(z)&  1+\mu  a_{44}(z)
\end{pmatrix}   , \quad  \quad
A_{21}(z):=\det\begin{pmatrix}
\mu  a_{23}(z)&\mu  a_{24}(z) \\
 \mu  a_{34}(z)& 1+\mu  a_{44}(z)
\end{pmatrix}  , \\
&A_{13}(z):=\det\begin{pmatrix}
\lambda  a_{23}(z)& 1+\mu a_{33}(z)\\
\lambda  a_{24}(z)& \mu  a_{34}(z)
\end{pmatrix}   , \quad  \quad
A_{31}(z):=\det\begin{pmatrix}
\mu  a_{23}(z)&\mu  a_{24}(z) \\
 1+\mu a_{33}(z)&\mu  c_{34}(z)
\end{pmatrix}  , \\
&A_{23}(z):=\det\begin{pmatrix}
1+\lambda  a_{22}(z)&\mu  a_{23}(z) \\
\lambda a_{24}(z)& \mu  a_{34}(z)
\end{pmatrix}   , \quad  \quad
A_{32}(z):=\det\begin{pmatrix}
1+\lambda  a_{22}(z)&\mu  a_{24}(z) \\
\lambda  a_{23}(z) &\mu  a_{34}(z)
\end{pmatrix}
\end{array}
\end{equation}
and
\begin{equation}\label{def:Aii}
\begin{array}{cl}
A_{11}(z):=&\det\begin{pmatrix}
1+\mu a_{33}(z)&\mu  a_{34}(z)\\
\mu  a_{34}(z)& 1+\mu  a_{44}(z)
\end{pmatrix} , \\
A_{22}(z):=&\det\begin{pmatrix}
1+\lambda  a_{22}(z)&\mu  a_{24}(z) \\
\lambda    a_{24}(z)& 1+\mu  a_{44}(z)
\end{pmatrix},\\
A_{33}(z):=&\det\begin{pmatrix}
1+\lambda  a_{22}(z)&\mu  a_{23}(z) \\
\lambda  a_{23}(z)& 1+\mu a_{33}(z)
\end{pmatrix}.
\end{array}
\end{equation}
The definitions \eqref{def:Aij} imply that
\begin{equation}\label{B_iB_j}
\mu  A_{12}(z)=\lambda  A_{21}(z) , \quad \mu  A_{13}(z)=\lambda
A_{31}(z), \quad   A_{23}(z)=  A_{32}(z)
\end{equation}
By combining \eqref{def:Aij} and \eqref{def:Aii} one derives
\begin{equation}\label{AB_1}
\begin{array}{cl}
&A_{11}(z)A_{33}(z)-A_{13}(z)A_{31}(z)=\Delta^\mathrm{ees}_{0 \lambda \mu }(z)\cdot [1+\mu a_{33}(z)], \\
&A_{11}(z)A_{22}(z)-A_{12}(z)A_{21}(z)=\Delta^\mathrm{ees}_{0 \lambda \mu }(z)\cdot [1+\mu  a_{44}(z)], \\
&A_{22}(z)A_{33}(z)-A_{23}(z)A_{32}(z)=\Delta^\mathrm{ees}_{0
\lambda \mu }(z)\cdot [1+\lambda a_{22}(z)]
\end{array}
\end{equation}
and
\begin{equation}\label{new_ineq_1}
A_{11}(z)A_{22}(z)A_{33}(z)-A_{12}(z) A_{23}(z)
A_{31}(z)=\Delta^\mathrm{ees}_{0 \lambda \mu }(z)R_{\lambda\mu}(z),
\end{equation}
where
$$R_{\lambda\mu}(z)= [1+\lambda a_{22}(z)]\cdot [1+\mu a_{33}(z)] \cdot [1+\mu  a_{44}(z)]-\lambda\mu^2 a_{13}(z)a_{14}(z)a_{24}(z).
$$
Then the equality $\Delta^\mathrm{ees}_{0 \lambda \mu }(z_{31})=0$
and identities  \eqref{AB_1} imply that
\begin{equation}\label{AB_2}
\begin{array}{cl}
&A_{11}(z_{31})A_{33}(z_{31})=A_{13}(z_{31})A_{31}(z_{31}), \\
&A_{11}(z_{31})A_{22}(z_{31})=A_{12}(z_{31})A_{21}(z_{31}), \\
&A_{22}(z_{31})A_{33}(z_{31})=A_{23}(z_{31})A_{32}(z_{31})
\end{array}
\end{equation}
Similarly, by \eqref{new_ineq_1} from  $\Delta^\mathrm{ees}_{0
\lambda \mu }(z_{31})=0$ it follows that
\begin{equation}\label{new_ineq_2}
A_{11}(z_{31})A_{22}(z_{31})A_{33}(z_{31})=A_{12}(z_{31})
A_{23}(z_{31}) A_{31}(z_{31}).
\end{equation}

We note that the functions $1+\mu  a_{33}(z)$ and $1+\mu  a_{44}(z)$
are strictly decreasing in $(-\infty,0)$. Then the relations
$z_{31}<z_{21}$ and  \eqref{ineq5} yield
\begin{equation}\label{niq1}
1+\mu  a_{33}(z_{31})>1+\mu  a_{33}(z_{21})>0  \quad \text{and}
\quad   1+\mu  a_{44}(z_{31})>1+\mu  a_{44}(z_{21})>0 .
\end{equation}
Since $\lambda,\mu<0$ and the functions $a_{23}$, $a_{24}$, $a_{34}$
are positive-valued (see Proposition \ref{prop:asymp_functions}),
from \eqref{niq1} it follows that
\begin{eqnarray}\label{ineqB12}
\nonumber
&A_{12}(z_{31})=\lambda  a_{23}(z_{31})\big[1+\mu  a_{44}(z_{31})\big]-\lambda\mu a_{24}(z_{31})a_{34}(z_{31})<0, \\
&A_{13}(z_{31})=\lambda\mu  a_{23}(z_{31})a_{34}(z_{31})-\lambda
a_{24}(z_{31})\big[ 1+\mu a_{33}(z_{31})\big]>0.
\end{eqnarray}
The relations   \eqref{B_iB_j}, \eqref{AB_2} and  \eqref{ineqB12}
imply
\begin{equation}\label{AB_3}
\begin{array}{cl}
&A_{11}(z_{31})A_{22}(z_{31})=A_{12}(z_{31})A_{21}(z_{31})=\frac{\mu
A^2_{12}(z_{31})}{\lambda}>0, \\
&A_{11}(z_{31})A_{33}(z_{31})=A_{13}(z_{31})A_{31}(z_{31})=\frac{\lambda
A^2_{31}(z_{31})}{\mu}> 0,
\end{array}
\end{equation}
which means that the numbers $A_{11}(z_{31}),\, A_{22}(z_{31}),\,
A_{33}(z_{31})$ have the same signs. Meanwhile, the relations
\eqref{2-ineq} and \eqref{loc_roots2}  yield
$A_{11}(z_{31})=\Delta^\mathrm{ees}_{00\mu}(z_{31})> 0$. Then by
\eqref{AB_3} one then concludes that
\begin{equation}\label{pos_A123}
A_{22}(z_{31})> 0 \text{\,\, and \,\,}  A_{33}(z_{31})> 0.
\end{equation}
Similarly, from \eqref{new_ineq_2}, \eqref{ineqB12} and
\eqref{pos_A123} it follows that
\begin{equation}
\label{A23l0} A_{23}(z_{31})<0.
\end{equation}
In its turn, \eqref{B_iB_j}, \eqref{ineqB12} and \eqref{A23l0} imply
that
\begin{equation}\label{signs_Bij}
\begin{array}{ll}
A_{12}(z_{31})<0, & A_{21}(z_{31})<0, \\
A_{23}(z_{31})<0, & A_{32}(z_{31})<0, \\
A_{31}(z_{31})>0, & A_{13}(z_{31})>0.
\end{array}
\end{equation}
Meanwhile, the equalities \eqref{B_iB_j} and  \eqref{AB_2} give that
\begin{equation}\label{AB_7}
\begin{array}{ll}
A_{12}^2(z_{31})=\tfrac{\lambda A_{11}(z_{31})A_{22}(z_{31})}{\mu},
&\quad A_{21}^2(z_{31})=\tfrac{\mu
A_{11}(z_{31})A_{22}(z_{31})}{\lambda}
\\
A_{23}^2(z_{31})= A_{22}(z_{31})A_{33}(z_{31}),  &\quad
A_{32}^2(z_{31})= A_{22}(z_{31})A_{33}(z_{31})
\\
A_{31}^2(z_{31})=\tfrac{\mu A_{11}(z_{31})A_{33}(z_{31})}{\lambda},
&\quad A_{13}^2(z_{31})=\tfrac{\lambda
A_{11}(z_{31})A_{33}(z_{31})}{\mu}.
\end{array}
\end{equation}
Taking into account the signs \eqref{signs_Bij} of the numbers
$A_{ij}(z_{31})$, $i,j=1,2,3,$ $i\neq j$, by \eqref{AB_7}  one
arrives at
\begin{equation}\label{AB_6}
\begin{array}{ll}
A_{12}(z_{31})= -\sqrt{\tfrac{\lambda
A_{11}(z_{31})A_{22}(z_{31})}{\mu}}  , \,\,     &\quad
A_{21}(z_{31})=-\sqrt{\tfrac{\mu
A_{11}(z_{31})A_{22}(z_{31})}{\lambda}}
,\\
A_{23}(z_{31})=-\sqrt{ A_{22}(z_{31})A_{33}(z_{31})}, \,\,  &\quad
A_{32}(z_{31})= -\sqrt{ A_{22}(z_{31})A_{33}(z_{31})}
, \\
A_{31}(z_{31})=\sqrt{\tfrac{\mu
A_{11}(z_{31})A_{33}(z_{31})}{\lambda}}, \,\, &\quad A_{13}(z_{31})=
\sqrt{\tfrac{\lambda A_{11}(z_{31})A_{33}(z_{31})}{\mu}}.
\end{array}
\end{equation}

Now let us expand the determinant \eqref{determinant-ees} along its
first column and obtain:
\begin{align*}
\Delta^\mathrm{ees}_{\gamma \lambda \mu }(z)
=& [1+2\gamma  a_{11}(z)]\Delta^\mathrm{ees}_{0 \lambda \mu }(z)\\
&-2\gamma  a_{12}(z)
\begin{vmatrix}
\lambda  a_{12}(z)& \mu  a_{13}(z)& \mu  a_{14}(z) \\
\lambda  a_{23}(z)& 1+\mu  a_{33}(z)&\mu  a_{34}(z) \\
\lambda  a_{24}(z)& \mu  a_{34}(z)&1+\mu  a_{44}(z)
\end{vmatrix}\\
&+2\gamma a_{13}(z)\begin{vmatrix}
\lambda  a_{12}(z)&\mu a_{13}(z)&\mu  a_{14}(z)\\
1+\lambda  a_{22}(z)& \mu  a_{23}(z)& \mu  a_{24}(z) \\
\lambda  a_{24}(z)& \mu  a_{34}(z)&1+\mu  a_{44}(z)
\end{vmatrix}\\
&-2\gamma a_{14}(z)  \begin{vmatrix}
\lambda  a_{12}(z)&\mu a_{13}(z)&\mu  a_{14}(z)\\
1+\lambda  a_{22}(z)& \mu  a_{23}(z)& \mu  a_{24}(z) \\
\lambda  a_{23}(z)& 1+\mu  a_{33}(z)&\mu  a_{34}(z)
\end{vmatrix}\\
=&[1+2\gamma  a_{11}(z)]\Delta^\mathrm{ees}_{0 \lambda \mu }(z)
-2\gamma\lambda a^2_{12}(z)A_{11}(z)+2\gamma\mu a_{12}(z)a_{13}(z)A_{12}(z)\\
&-2\gamma\mu a_{12}(z)a_{14}(z)A_{13}(z)+2\gamma\lambda
a_{13}(z)a_{12}(z)A_{21}(z)
-2\gamma\mu a^2_{13}(z)A_{22}(z)\\
&+2\gamma\mu a_{13}(z)a_{14}(z)A_{23}(z)-2\gamma\lambda a_{14}(z)a_{12}(z)A_{31}(z)\\
&+2\gamma\mu a_{14}(z)a_{13}(z)A_{32}(z)-2\gamma\mu
a^2_{14}(z)A_{33}(z). \nonumber
\end{align*}
By taking into account the equalities
$\Delta^\mathrm{ees}_{0\lambda\mu}(z_{31}) = 0$ and \eqref{AB_6}, we
derive from this representation of
$\Delta^\mathrm{ees}_\gamma\lambda\mu(z)$ that
\begin{align}\label{D_z_31<0}
\Delta^\mathrm{ees}_{\gamma \lambda \mu }(z_{31})=&-2\gamma\lambda
a^2_{12}(z_{31})A_{11}(z_{31})\\ \nonumber & -2\gamma\mu
a^2_{13}(z_{31})A_{22}(z_{31}) -2\gamma\mu
a^2_{14}(z_{31})A_{33}(z_{31}) \\  \nonumber &+2\gamma\mu
a_{12}(z_{31})a_{13}(z_{31})A_{12}(z_{31})-2\gamma\mu
a_{12}(z_{31})a_{14}(z_{31})A_{13}(z_{31})\\ \nonumber
&+2\gamma\lambda a_{12}(z_{31})a_{13}(z_{31})A_{21}(z_{31})
+2\gamma\mu a_{13}(z_{31})a_{14}(z_{31})A_{23}(z_{31})\\\nonumber
&-2\gamma\lambda a_{12}(z)a_{14}(z_{31})A_{31}(z_{31}) +2\gamma\mu
a_{13}(z)a_{14}(z_{31})A_{32}(z_{31})
\\  \nonumber &-2\gamma\lambda
a^2_{12}(z_{31})A_{11}(z_{31})-2\gamma\mu
a^2_{13}(z_{31})A_{22}(z_{31})\\  \nonumber &-2\gamma\mu
a^2_{14}(z_{31})A_{33}(z_{31}) -4\gamma\sqrt{\lambda\mu
A_{11}(z_{31}) A_{33}(z_{31})} a_{12}(z_{31})a_{14}(z_{31})\\
\nonumber &-4\gamma\mu a_{13}(z_{31})a_{14}(z_{31})\sqrt{
A_{22}(z_{31}) A_{33}(z_{31})}\\ \nonumber
&-4\gamma \sqrt{\lambda\mu A_{11}(z_{31}) A_{22}(z_{31})} a_{12}(z_{31})a_{13}(z_{31})\\
=-2\Big(a_{12} & (z_{31})\sqrt{\gamma\lambda A_{11}(z_{31})}+
a_{13}(z_{31})\sqrt{\gamma\mu A_{22}(z_{31})}+
a_{14}(z_{31})\sqrt{\gamma\mu A_{33}(z_{31})}\Big)^2. \nonumber
\end{align}
In view of \eqref{pos_A123} and strict positivity of the numbers
$a_{12}(z_{31})$, $a_{13}(z_{31})$, $a_{14}(z_{31})$, the last
equality immediately yields
\begin{equation}\label{D4in1}
\Delta^\mathrm{ees}_{\gamma \lambda \mu }(z_{31})<0.
\end{equation}

Combining the assertions \eqref{2-ineq} with inequalities   $
z_{21}< z_{32}< z_{22} $  and   $ z_{22}< z_{33} $ one concludes
that
$$
A_{11}(z_{32})< 0, \quad  A_{22}(z_{32})<0, \quad  A_{33}(z_{32})<0
$$
and
$$
A_{11}(z_{33})> 0, \quad  A_{22}(z_{33})> 0, \quad  A_{33}(z_{33})>
0.
$$
respectively. In a way similar to the one we passed to show
\eqref{D4in1}, one then proves that
\begin{align*}
\Delta^\mathrm{ees}_{\gamma \lambda \mu }(z_{32})=&
\,2\Big(a_{12}(z_{32})\sqrt{-\gamma\lambda
A_{11}(z_{32})}\\
&+a_{13}(z_{32})\sqrt{-\gamma\mu
A_{22}(z_{32})}+a_{14}(z_{32})\sqrt{-\gamma\mu
A_{33}(z_{32})}\Big)^2\\
>&\, 0
\end{align*}
and
\begin{align*}
\Delta^\mathrm{ees}_{\gamma \lambda \mu
}(z_{33})=&\,-2\Big(a_{12}(z_{33})\sqrt{\gamma\lambda
A_{11}(z_{33})}\\
&+a_{13}(z_{33})\sqrt{\gamma\mu
A_{22}(z_{33})}+a_{14}(z_{33})\sqrt{\gamma\mu
A_{33}(z_{33})}\Big)^2\\
<&\, 0.
\end{align*}

According to \eqref{ineq_glm1} we have
$Q^{-}(\gamma,\lambda,\mu)>0$. Then by Lemma \ref{lem:asymp_det} we
have
$$
\lim\limits_{z\rightarrow -\infty}\Delta^\mathrm{ees}_{\gamma
\lambda \mu }(z)=1
 \quad \text{and} \quad
 \lim\limits_{z\nearrow 0}\Delta^\mathrm{ees}_{\gamma \lambda \mu }(z)=+\infty.
$$
Therefore
\begin{align*}
&\lim\limits_{z\rightarrow -\infty}\Delta^\mathrm{ees}_{\gamma \lambda \mu }(z)=1  , \quad  \Delta^\mathrm{ees}_{\gamma \lambda \mu }(z_{31})< 0, \quad \Delta^\mathrm{ees}_{\gamma \lambda \mu }(z_{32})> 0,\\
& \quad \quad  \quad \Delta^\mathrm{ees}_{\gamma \lambda \mu
}(z_{33})< 0, \quad \lim\limits_{z\nearrow
0}\Delta^\mathrm{ees}_{\gamma \lambda \mu }(z)=+\infty.
\end{align*}
The above relations yield the existence of four zeros  $z_{41}$,
$z_{42}$, $z_{43}$ and $z_{44}$ of the function
$\Delta^\mathrm{ees}_{\gamma \lambda \mu }(z)$, satisfying the
following inequalities
\begin{equation}\label{loc_roots3}
z_{41}< z_{31}< z_{42}< z_{32}< z_{43}< z_{33}< z_{44}<0.
\end{equation}
Moreover, by  Lemma \ref{lem:det_zeros_vs_eigen} the function
$\Delta^\mathrm{ees}_{\gamma \lambda \mu }(z)$ has no other zeros
lying in $(-\infty,0)$. Thus, the same lemma also implies that the
operator $H^\mathrm{ees}_{\gamma \lambda \mu }(0)$ has exactly four
different bound states with the energies lying below the essential
spectrum. \qed
\medskip

\noindent\textit{Proof of Theorem \ref{teo:bound_K} }.  Let
$\gamma,\lambda,\mu\in \mathbb{R}$, $ m, n=0,1,2,3,4,5,6,7$ and $(
\gamma,\lambda,\mu)\in \mathbb{G}_{mn}$. Corollary
\ref{teo:eig_of_H(0)} states that  the operator $H_{
\gamma\lambda\mu}(0)$ has exactly $m$ bound states with energies
lying  below the essential spectrum   and exactly $n$  bound states
with energies lying above the essential spectrum. In turn,  Theorem
\ref{teo:disc_Kvs0} guarantees that the operator
$H_{\gamma\lambda\mu}(K)$ has at least $m$ bound states with
energies below the essential spectrum and at least $n$ bound states
with energies above the essential spectrum.  Moreover, the equality
\eqref{moment_poten} implies that  the operator
$V_{\gamma\lambda\mu}$ has rank at most seven. Therefore, by the
minimax principle (see \cite{RSimon:IV}, page 85) the operator
$H_{\gamma\lambda\mu}(K)$ has at most seven eigenvalues. This yields
that, for $( \gamma,\lambda,\mu)\in \mathbb{G}_{mn} $ and
$m+n=7$, the operator $H_{\gamma\lambda\mu}(K)$ has exactly $m$
bound states with energies below the essential spectrum and exactly
$n$ bound states with energies above the essential spectrum. This
completes the proof.\hfill $\square$

\end{document}